%% file: main.tex
\documentclass[10pt]{article}
\usepackage[preprint]{tmlr}
\usepackage{amsmath,amssymb,amsthm,graphicx,booktabs,longtable,tabularx,array,calc}
\usepackage{hyperref,url}
\usepackage{algorithm,algpseudocode}
\hypersetup{colorlinks=true,linkcolor=blue,citecolor=blue,urlcolor=blue,hypertexnames=false,pdftitle={Reusing Statistical Guarantees Under Change: Selection Costs and Event Ownership},pdfauthor={Ao Li and Weitong Chen}}

\newtheorem{proposition}{Proposition}
\newtheorem{theorem}[proposition]{Theorem}
\title{Reusing Statistical Guarantees Under Change:\\Selection Costs and Event Ownership\thanks{Codex was used to assist with experimental code development, Google Gemini was used to assist with manuscript writing, and Elicit was used for additional literature search and citation checking. The authors are responsible for all content.}}
\author{\name Ao Li \\
\addr Australian Institute for Machine Learning, Adelaide University,\\Adelaide, Australia
\AND
\name Weitong Chen \email weitong.chen@adelaide.edu.au \\
\addr Australian Institute for Machine Learning, Adelaide University,\\Adelaide, SA 5005, Australia}
\begin{document}
\maketitle
\begin{abstract}
An adaptive system can retain an earlier estimate while changing what that estimate is used to guarantee. The guarantee depends on the experiment that produced the evidence, including how its target was selected and which assumptions connect it to the current claim. We show that past-measurable selection preserves future martingale dynamics without preserving a unit expected starting value. An exact finite counterexample, using a strictly positive family, refutes a published changing-prior extension as stated with probability one and rules out unrestricted bounds based only on switch count and the error level. We give a paid-start multi-switch correction and establish $m$ as the sharp worst-case expected starting-moment cost for an unstructured family of $m$ hypotheses. We then develop Adaptive Assurance, a framework in which statistical events retain their experiment, funding and outcome bindings while claims and premise dependencies change. Compatible-outcome projection supports reuse, selective invalidation, fresh acquisition and auditable decisions. Two histories can have identical current regions yet require different responses to a later edit, making evidence origin part of the assurance state. Studies in six public-recorded domains exercise heterogeneous outcomes and access contracts. Component interventions expose erroneous inheritance and unnecessary resets. The framework separates conditional guarantee recovery from competence adaptation, with explicit costs for evidence, computation and verification.
\end{abstract}

\section{Introduction}
A monitoring system audits a simulator's remaining-useful-life outputs. Later, a revised reporting convention clips those outputs at a different value. The old observations and their confidence interval are still stored. Some newly reported outputs have also been inspected directly. When the relation between the two conventions is withdrawn, which evidence can still justify a claim about the revised output? Direct target observations remain relevant. Values derived through the withdrawn relation lose their support. Retaining the old interval alone cannot express this difference.

The same tension arises when a classifier is replaced, a health-data request changes from disease presence to a finer recorded outcome, or an analyzer reference becomes unavailable. A stored number can remain informative without retaining the guarantee required by its new use. Reliable reuse therefore requires distinguishing the experiment that established validity from the claims that now consume it.

Selection makes this distinction mathematically consequential. In the changing-prior extension of \citet{chugg2023}, a prior chosen from past observations is applied to an accumulated inferential process with a boundary based on the number of switches. We construct a positive martingale family that violates the printed one-switch boundary on every outcome. Each fixed-parameter process remains correctly initialized, and selection preserves its future dynamics. The missing quantity is the selected mixture's expected starting wealth. Established density and joint-moment methods price that dependence; future-only rebasing funds a different experiment. The result identifies a concrete cost of reassigning ownership to already accumulated evidence. A corrected multi-switch bound pays the observed start at each pre-observation birth, yielding native events that the framework can retain and compose.

That boundary motivates a positive construction. Adaptive Assurance keeps a native event's statistic, experiment start, outcome bindings and failure allocation together. A changed claim consumes the event through qualified semantic relations. New evidence receives its own owner, and premise withdrawal removes only dependent constraints. A compatible-outcome set then supplies the range needed for a justified decision. The maintained state must preserve more than its current range: two histories can define the same full feasible region yet respond differently to the same source revocation. \emph{An assurance state must preserve not only what is currently possible, but why it is possible.}

This gives three linked design principles. \textbf{Validity has an owner}: selection and funding belong to the native experiment. \textbf{Assurance state is edit-relative}: retained origin and dependencies determine lawful future responses. \textbf{Recovery cost is distinct from validity}: correct reuse can reduce reacquisition, but a valid stop rule need not choose an economical audit. These principles organize the theory, method and experiments below.

The paper develops this invariant at three levels. The selected-start theorem establishes an impossibility boundary for unpaid reassignment of accumulated evidence. An event-owned assurance representation and edit lifecycle preserve statistical ownership together with premise support. Executable heterogeneous maintenance and six public-recorded domains demonstrate the consequences for reuse, invalidation and recovery. Established statistical inference, truth maintenance and optimisation supply the construction's foundations \citep{entrust2018,dekleer1986,ppi2023,evergreen2026}. Component interventions and a controlled provenance comparison show why the ownership fields matter; matched acquisition controls locate the remaining recovery costs. Section~\ref{sec:prior} relates this contract to its nearest methods.

\input{sections/selected_start}

\section{From an earlier guarantee to an editable assurance state}\label{sec:contract}
\subsection{A running lifecycle}
Return to the simulator example. Let $R\geq0$ be a recorded remaining-useful-life output, and define the old and revised quantities
\[
 x=\min(R,200)/200,\qquad y=\min(R,125)/125.
\]
An audit of $x$ owns an interval for its original mean. A qualified clipping relation constrains $y$ without turning that interval into a direct audit of $y$. Individually observed $x$ values also yield exact derived values $y=\min(1,8x/5)$. Independently observed $y$ values have a different origin. Withdrawing the relation removes the derivations while preserving the independent observations. Acquisition can then fill the exposed gap. Figure~\ref{fig:ownership} shows this lifecycle alongside the generic contract.

\begin{figure}[t]
\centering\includegraphics[width=\linewidth]{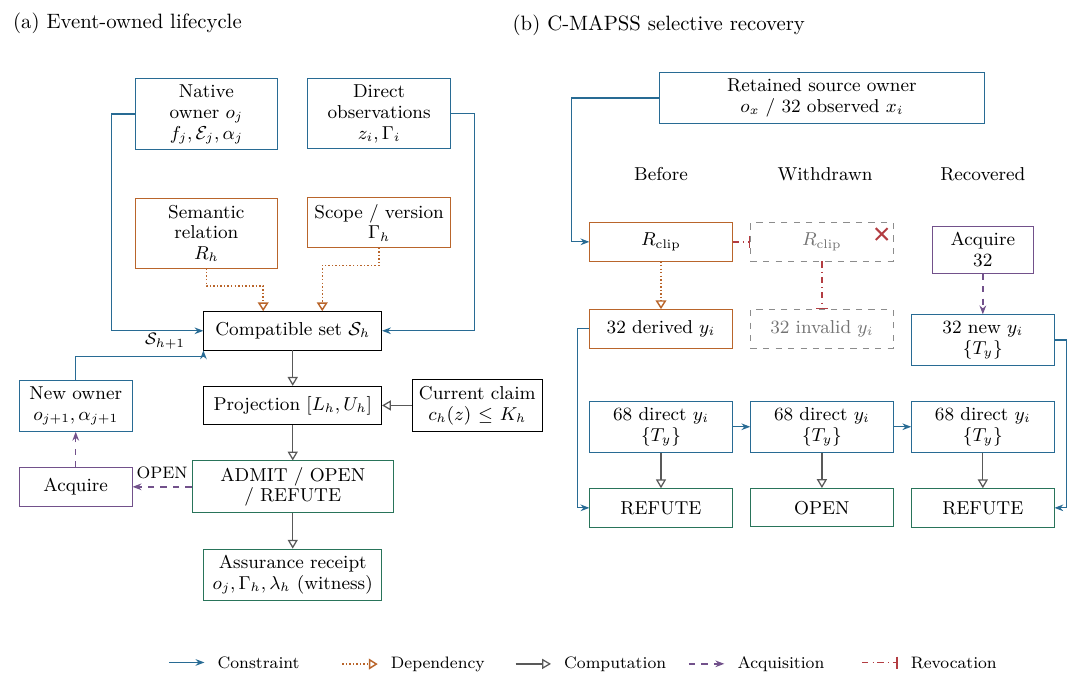}
\caption{\textbf{Event-owned assurance.} (a) Constraints and dependencies define the compatible set; the set and current claim determine the projected range. OPEN initiates acquisition and a new funded owner. Edge styles distinguish these operations. (b) Withdrawing the clipping relation invalidates derived target evidence while retaining the source owner and independent observations. Decisions progress from REFUTE to OPEN and back to REFUTE after selective recovery.}\label{fig:ownership}
\end{figure}

\subsection{Validity has an owner}
Let $z$ collect outcome coordinates bound to a subject, channel and version, with declared domain $D$. Binary losses, categorical outcomes and bounded measurements are different coordinate types. Unknown outcomes remain free coordinates. A native owner records
\[
 o_j=(f_j,\mathcal E_j,\beta_j,\alpha_j,b_j,\Gamma_j),
\]
where $f_j$ is the statistic, $\mathcal E_j$ its experiment, $\beta_j$ its outcome binding, $\alpha_j$ its failure allocation, $b_j$ its birth history, and $\Gamma_j$ its premises. The native event is $E_j=\{f_j(z)\in I_j\}$. The funding contract is
\begin{equation}\label{eq:funding}
 \Pr(E_j^c\mid\mathcal H_{b_j})\leq\alpha_j,
 \qquad \sum_{j\text{ created}}\alpha_j\leq\delta\quad\text{pathwise}.
\end{equation}
Earlier data may select and initialize an owner, but their selection cost must be accounted for at birth. Theorem~\ref{prop:paidstart} supplies a concrete instance: normalize the selected mixture by its observed start $S_i$ at $T_i$, then fund its future trajectory with $\alpha_i$. The original statistic and binding remain fixed. Its native event $E_i$ therefore satisfies the same conditional contract as an audit interval and can be used in Proposition~\ref{prop:reuse}. Separately funded fixed looks and a valid confidence sequence are alternative native components; the latter can own an entire observation trajectory.

Point evidence records an observed outcome and its origin. Qualified relations, such as a category coarsening, subset, unit conversion or clipping map, carry premises. Active objects at history $h$ define
\begin{equation}\label{eq:region}
 \mathcal S_h=\{z\in D:\ f_j(z)\in I_j\text{ for active owners};\
 z\text{ satisfies active points and qualified relations}\}.
\end{equation}
For a possibly history-selected claim $c_h(z)\leq K_h$, compute an outer range $[L_h,U_h]\supseteq c_h(\mathcal S_h)$. The output is ADMIT if $U_h\leq K_h$, REFUTE if $L_h>K_h$, and OPEN otherwise. Unqualified scope or failed access also produces OPEN.

\begin{proposition}[Adaptive consumers of funded events]\label{prop:reuse}
Assume Equation~\ref{eq:funding}, truthful active points, and valid active semantic and scope premises. With probability at least $1-\delta$, every emitted admission or refutation, across arbitrarily selected consumers of the created native events, is correct for the bound outcomes.
\end{proposition}
A union bound over the conditionally funded events establishes simultaneous validity. On that event, the true outcome lies in every active compatible set, including after constraints are removed. Selecting a consumer changes its image while preserving this membership. This is confidence-set projection applied to an editable event contract; Appendix~\ref{app:projection} gives the countable conditional argument. If $B$ is failure of a truth, semantic or scope premise required by a claim, the unconditional error is bounded by $\delta+\Pr(B)$ under the native-event assumptions. Premise risk is a separate quantity, potentially unknown.

\subsection{Edits follow evidence dependencies}
A claim change modifies the consumer. A semantic revision modifies the relation connecting outcome channels. Erasing a point removes that observation constraint while an independently funded aggregate can remain. Revoking a premise removes every dependent object, including inferred points. A new measurement receives its own binding. A new statistical experiment consumes fresh allocation. Spent allocation remains spent after deletion or revocation.

These operations make selective recovery possible. When the clipping relation is withdrawn, old-channel evidence remains valid for $x$ and independently observed values remain valid for $y$. Only the relation-dependent target values need reacquisition. The original owner is retained without being redefined as a new target audit.

\subsection{Assurance state is edit-relative}
Belief revision and base-generated contraction provide the precedent: a logical closure and the supporting base encode different aspects of a belief state; equal closures can retain different change responses \citep{agm1985,hansson2007,oveisi2014}. Our assurance state extends this support distinction to funded events, outcome bindings and semantic premises, so that later edits determine both the current decision and the evidence requiring reacquisition.

Current correctness alone does not specify an editable assurance state. Consider $x,y\in[0,1]$ with relation $x=y$. One history contains direct $x=0$ evidence supported by source $T_x$; another contains direct $y=0$ evidence supported by $T_y$. Both define $\mathcal S_h=\{(0,0)\}$ and admit $y\leq1/2$. Revoking $T_x$ makes the first history OPEN and leaves the second admission intact.

\begin{proposition}[Current-region equality does not preserve edits]\label{prop:edit}
A representation determined only by $\mathcal S_h$ cannot reproduce every provenance-sensitive premise-revocation response of the event-owned contract.
\end{proposition}
Identical initial representations and an identical edit require identical deterministic outputs, but the justified responses differ. Transforming the current region cannot restore discarded origin information. Section~\ref{sec:ablation} gives an exact counterpart using C-MAPSS records.

Fix an admitted language $\mathcal U$ of edits, claims and access responses. Histories are response-equivalent when every future sequence in $\mathcal U$ yields the same required outputs, including refusal and availability. An assurance representation must distinguish inequivalent histories. This standard future-response criterion identifies the distinctions needed by the particular lifecycle, and assumption-based support maintenance supplies its provenance machinery \citep{dekleer1986}.

The admitted language changes the required state. For one fixed funded Bernoulli inventory, every upper-threshold response after named erasures depends only on retained-zero membership. Unknown and retained-one cells share the conservative completion response. The resulting $2^N$ response classes expand to $3^N$ when honest retained-versus-paid-reacquisition availability is included. The corresponding information requirements are $N$ and $N\log_2 3$ bits. Appendix~\ref{app:editlanguage} proves these restricted statements. Physical storage also contains scope, protocol and funding information; response-class cardinality describes distinguishability rather than exponential heap size.

\subsection{Changed claims can require richer outcome information}
Presence vectors $(1,1)$ can describe original coded outcomes $(1,1)$ or $(4,4)$, whose normalized ordinal means differ. Conversely, numeric pairs $(0,1)$ and $(1/2,1/2)$ share a mean but have different exceedance rates. These aliases determine which consumers a retained object can support, independently of predictive fit.

Finite categories therefore use simplex coordinates $z_{ik}$ with qualified coarse-to-fine mappings. Presence fixes $\sum_{k>0}z_{ik}$ while preserving alternative positive categories; ordinal claims use linear scores. Bounded measurements retain their numeric coordinates. For the running example, the outer clipping constraints $y\geq x$, $y\leq8x/5$, $y\leq1$ contain every actual pair. Revocation disconnects the two channels.

The executable class uses bounded linear claims over numeric boxes, categorical simplices and linear outer relations. Mature convex relaxation provides containment and can widen OPEN regions. This gives one heterogeneous implementation of the contract, with expressivity and approximation costs stated explicitly.

\section{Projection, evidence acquisition and two-layer trust}\label{sec:method}
\subsection{Numerical projection and exact verification}
Store the active region as $Az\leq b$, $\ell\leq z\leq u$, with rational coefficients. Linear programming proposes nonnegative multipliers $\lambda$ for a claim $d^\top z+q$. Independent arithmetic verifies
\begin{equation}\label{eq:dual}
 U=q+\lambda^\top b+\sum_r\max\{r_r\ell_r,r_ru_r\},
 \qquad r=d-A^\top\lambda.
\end{equation}
Applying the same construction to $-d$ gives $L$. Exact residuals preserve weak duality even when proposed multipliers are approximate. A failed numerical solve returns the box bound. The assurance receipt contains the claim, outcome bindings, owners, premises, edit history, active constraints and bound witnesses; replay requires neither new outcomes nor a solver.

The binary specialization groups subjects with identical active-event and target coefficients into count cells. Every feasible integer cell count has a realizing subject assignment, so linear objective ranges are preserved. Access or provenance distinctions require refinement when an admitted edit can expose them.

\subsection{Premise trust and implication validity}
The mechanism separates two proof obligations. The \emph{premise layer} establishes that observations, native intervals, sampling, semantics and source bindings are authentic and valid. The \emph{implication layer} establishes what follows from those represented premises. Replay owns the latter obligation and exposes its dependence on the former.

\begin{proposition}[Conditional receipt replay]\label{prop:replay}
If replay verifies the nonnegative multipliers and exact residual bounds of a receipt, its reported range contains the represented claim on the active represented region. If authenticated, valid premises place the true bound outcome in that region, every qualified threshold decision follows for that outcome.
\end{proposition}
The first statement is rational weak duality; the second is set membership followed by projection. Before accepting a native row, replay also checks that the funded statistic and subject, channel and version bindings remain unchanged. Hashes bind document content for integrity checking. Authentication, sampling honesty and native interval generation belong to the premise layer.

A controlled forgery makes this separation concrete. An asserted interval $z\in[.4,.6]$ refutes $z\leq.3$. Replacing the entire interval and its witness by a coherent document asserting $z\in[.1,.2]$ passes implication replay and admits the claim, despite diagnostic external truth $z=.5$. The forged document fails the premise obligation. Changing only the reported bound while retaining the original premises is rejected. Thus trustworthiness consists of a checkable implication together with visible, separately warranted premises.

\subsection{Recovery cost is distinct from validity}
\begin{algorithm}[t]
\caption{Event-owned assurance maintenance}\label{alg:maintenance}
\begin{algorithmic}[1]
\Require Retained owners and points, current claim, qualified premises, access and remaining allocation
\State Apply edits; remove dependent constraints; retain spent allocations
\State Project the claim and replay its bound witness
\If{qualified outer range decides the threshold}
 \State \Return decision and assurance receipt
\EndIf
\If{scope, access or funding is unavailable}
 \State \Return OPEN with unresolved obligations
\EndIf
\State Fix an audit population before observing its outcomes
\State Create a native owner and fund its looks or sequential process
\State Acquire outcomes with origin and premise dependencies
\State Reproject at funded looks; decide or return OPEN at the declared limit
\end{algorithmic}
\end{algorithm}

Algorithm~\ref{alg:maintenance} separates a proposal about what to inspect from the compatible range that justifies stopping. The binary study compares endpoint-guided, uniform and fresh audit rules. Endpoint guidance selects observations that can separate opposite-answer witnesses. Each audit population is fixed before its draws and its declared looks are separately funded. Numeric and categorical studies use corresponding native audits to test assurance operations.

A valid current interval can still conceal different acquisition requirements. For $(y_1+y_2)/2\leq1/4$, retaining $y_1=y_2$ makes one inspection sufficient. Without the relation, a first observed zero leaves both answers possible, requiring two inspections in the worst case. Both initial count ranges are $[0,2]$. Decision-region and decision-tree methods solve acquisition under richer supplied hypotheses, tests and priors \citep{golovin2010,javdani2014}. The distinction between assurance state and acquisition state motivates recording evidence, search, verification and retained memory separately.

\section{Empirical consequences of event ownership}\label{sec:experiments}
\subsection{Domains and authority}
The domain instruments in Table~\ref{tab:domains} expose distinct edit and access operations. Public recorded outcomes are hidden from a controller until the declared inspection reveals them; their complete reference is available to the evaluator. Recorded labels and simulator outputs supply truth relative to these experiments. Source work, computational references and semantic relations are explicit parts of the authority and resource contract.

\begin{table}[t]\centering\small
\begin{tabularx}{\linewidth}{>{\raggedright\arraybackslash}p{.23\linewidth} r >{\raggedright\arraybackslash}X >{\raggedright\arraybackslash}X}\toprule
Frame & Size & Assurance operation & Reference / authority\\\midrule
Adult\newline\citep{adult1996} &16,281&Model replacement, selection, subgroup, erasure&32,561 training labels; recorded income\\
BPI2012\newline\citep{bpi2012} &8,897&Workflow-rule and scope changes&Computational Petri reference; 263 unknown traces excluded\\
Heart Disease\newline\citep{heart1989} &303/123/200&Presence to ordinal mean; relation/truth withdrawal&Recorded codes 0--4\\
C-MAPSS\newline\citep{cmapss2008} &100/259&Clipping revision; selective target recovery&Simulator-relative RUL\\
Seeds\newline\citep{seeds2010} &210&Unit conversion; unrelated-channel refusal&Recorded wheat-kernel measurements\\
Air Quality\newline\citep{air2008} &9,357&Available to incomplete full frame; source withdrawal&7,674 observed values; 1,683 arbitrary missing outcomes\\\bottomrule
\end{tabularx}
\caption{Six domains exercise one contract through different outcomes and operations. Sizes are fixed evaluation frames; repeated decisions and orders are dependent. Health codes are descriptive recorded outcomes, and the simulator and crop studies use recorded output and measurement contracts.}\label{tab:domains}\end{table}

Within each retention condition, deployable acquisition comparisons share source evidence, semantics, access, funding and claims. A complete census is charged. Exact enumeration is a small-case computational reference. Generic projection on the same active constraints is the mature matched primitive and reproduces the same finite ranges. Scope-blind copying and component omissions are invalid-input diagnostics. Appendix~\ref{app:experiments} gives authority passports, preprocessing and sampling units.

\subsection{Selective invalidation and recovery}
The three multivalued Heart Disease files \citep{heart1989} retain codes 0--4. A 64-subject binary presence audit supports an ordinal mean through the qualified coarse relation. Target audits and independent categorical points add finer information. Withdrawing the coarse relation preserves direct target evidence; withdrawing target truth removes it. The Hungarian processed file contains only binary outcomes, so requests for unavailable fine codes remain OPEN.

C-MAPSS \citep{cmapss2008,saxena2008} gives a numeric counterpart: revising a reporting convention changes the target without discarding the source experiment. Direct target observations retain independent support when the clipping relation is withdrawn. Selective recovery acquires only the exposed target values. Figure~\ref{fig:interval} separates this acquisition from premise withdrawal; Table~\ref{tab:receipt} gives the corresponding ranges. A future-version claim still requires a qualified link to that version.

\begin{figure}[t]\centering\includegraphics[width=.94\linewidth]{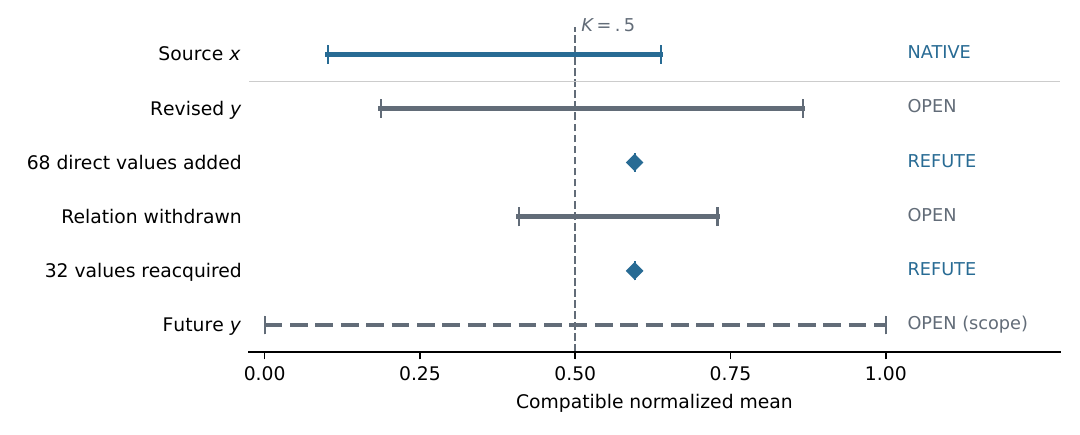}
\caption{\textbf{Assurance evolves with its evidence scope.} Normalized C-MAPSS FD001 means under the old ($x$) and revised ($y$) clipping conventions. Adding 68 direct values first completes the recorded target. Relation withdrawal then widens the range; selective reacquisition restores it. At the recorded target census it refutes $y\leq.5$; the future-scope request has no qualified binding. Source and target quantities are distinct estimands.}\label{fig:interval}\end{figure}
\input{tables/receipt}

Seeds \citep{seeds2010} uses a clipped area mean and a known unit conversion. Independent target observations survive source-area withdrawal; an unrelated perimeter claim lacks a supplied relation. Air Quality \citep{air2008} tests a different boundary: a uniform audit of 256 available analyzer values supports their clipped mean, but expanding to all hours adds 1,683 unconstrained missing outcomes. Even an available-reference census leaves
\[
 [S/N,(S+M)/N],\qquad S=\sum_{i\ \mathrm{available}}\min(1,\max(0,\mathrm{CO}_i/10)),\quad M=1683.
\]
This range covers every completion. Design-based sampling of the fixed available frame requires neither temporal independence nor missing-at-random assumptions. Revoking analyzer truth removes its points and aggregate. The data contain 9,357 dated rows, compared with the catalog count of 9,358.

The full-frame interval is the finite-record form of established worst-case partial-identification bounds for a bounded outcome with arbitrary missingness \citep{horowitz2000}. Its role here is to separate unavailable outcomes from sampling uncertainty and to retain the analyzer premise governing the observed contribution.

The heterogeneous study comprises 852 dependent decisions, with 405 qualified conclusions agreeing with recorded outcomes or all compatible missing completions. Its 291 full certificates replay without new outcomes. A separate binary relation/reference study contains 660 decisions, 364 qualified conclusions and 43 receipts. These counts describe implementation behavior. The reliability guarantee follows from the event contract and its premises. Exact finite controls compare 110 grouped binary projections with subject-level optimization and enumeration, and categorical/numeric controls check compatible-truth containment.

\subsection{Component interventions and a controlled provenance comparison}\label{sec:ablation}
To test the framework structure outside small witnesses, we apply controlled omissions to Heart, C-MAPSS, Seeds and Air frames. The native-identity intervention copies a complete source mean into a different target functional: presence into ordinal mean, one clipping convention into another, area into an unrelated perimeter channel, or available-frame mean into an incomplete full-frame claim. The complete contract rejects all four rebindings. Each domain has one contradictory copying decision on the predeclared five-threshold grid $K\in\{.1,.25,.5,.75,.9\}$ (Table~\ref{tab:ablation}). A contradiction is checked against recorded truth or an allowed completion of missing outcomes.

\input{tables/native_components}

A dependency omission instead keeps relation-derived evidence after relation withdrawal. Its numerical answer can remain equal to the recorded outcome even though the supporting premise has been removed. We therefore measure invalid evidence survival separately from contradictory decisions. With explicit dependencies, every derived target point is removed and every independent target point survives. Resetting all target evidence purchases those independent values again. At equal restored information, selective retention saves exposures. Construction, projection and replay remain separate costs.

The selective/reset proportions in these controlled edits are determined by how much independently supported evidence the intervention preserves. They test lifecycle semantics and recovery consequences; acquisition performance is assessed separately under common sampling contracts.

Scope and refusal have another effect. Dropping a future-scope requirement, or substituting the last qualified answer for OPEN, yields 19 unsupported decisions across 20 fixed requests. Such decisions need not contradict the old recorded frame; they lack a qualified link to the future claim. The interventions therefore distinguish contradictory decisions, stale support, scope failure and unnecessary reacquisition.

For the real-record edit analogue, two controlled C-MAPSS FD001 histories include the same full old-channel census. One obtains target values through the clipping relation; the other inspects all target values independently. Both define the identical joint singleton region and the exact target mean $1489/2500=.5956$. After the same relation withdrawal, their target ranges are respectively $[0,1]$ and $[.5956,.5956]$, giving OPEN and REFUTE for $y\leq.5$. This is a constructed acquisition-history comparison on real records, not a naturally observed deployment episode. It shows that the future-edit proposition applies to the actual heterogeneous quantities used by the method, even when the complete current region is known.

\subsection{Correct recovery does not guarantee economical acquisition}
Adult \citep{adult1996} uses source-trained logistic models with a replacement and threshold change. BPI2012 \citep{bpi2012} uses a supplied Petri reference qualified by checked firing paths and left-nullspace refutations \citep{murata1989}. Four fixed traces receive cost-based alignment checks \citep{adriansyah2011} through PM4Py \citep{pm4py2023}. Eight common source orders generate twelve trajectories per frame: three acquisition rules under summary, full, 128-point and no-new-access retention. Nine changes include a threshold, model or rule replacement, subgroup, common source-selected target, erasure, source-premise withdrawal and unsupported outcome relation. Each trajectory allocates .02 to its source and at most eight roots of .003 with separately funded looks.

\input{tables/lifecycle}
The binary study has 1,728 decisions across 192 trajectories and 1,227 qualified conclusions, all agreeing with the recorded frame. Endpoint guidance nevertheless buys more evidence than fresh evaluation in some matched regimes. Full retention can improve reuse without lowering total burden. Table~\ref{tab:lifecycle} retains those adverse comparisons; Appendix~\ref{app:burden} separates exposures from wall time. The source audits of 2,048 and 512 outcomes, training, reference construction, retained memory and replay remain additional resources. This evidence limits the acquisition policy while preserving the correctness and selective-invalidation contract.

A bounded supplementary study makes the statistical owner's contribution visible before census. Fixed-complement audits of revised C-MAPSS outputs and calendar-defined Air Quality consumers reach qualified decisions using partial samples (Appendix~\ref{app:strengthening}). Retention can reduce new observations, but its changed complement also changes sample composition and can delay stopping. An ordinary complement confidence interval intersecting all funded looks gives exactly the same decisions under the same information and funding. A 40-order paired study reports retention costs with median intervals and separate calendar consumers. Practical cached-interval and origin-blind-sample controls can inherit incorrect answers after declared consumer edits, whereas a provenance-tagged completion control preserves the same valid responses as the framework. Thus mature inference supplies the interval; the framework makes its ownership, editable dependencies and replayable implications explicit.

\section{Related work: inference, support and assurance lifecycles}\label{sec:prior}
\paragraph{Belief bases and partial identification.}
Belief-base contraction connects retained commitments, dependence and selective belief change \citep{oveisi2014}. Worst-case partial identification bounds outcomes under arbitrary missingness \citep{horowitz2000}. These precedents supply the support distinction and completion semantics; event ownership additionally retains the native statistical experiment, its allocation and outcome bindings through edits.

\paragraph{Dynamic assurance.}
ENTRUST combines design-time controller verification with runtime model updates, reconfiguration, reusable evidence and dynamically generated assurance cases \citep{entrust2018}. It supplies a mature surrounding architecture. Adaptive Assurance contributes an event-specific statistical ownership layer for reuse and invalidation under edited consumers and premises.

\paragraph{Truth maintenance and probability provenance.}
Assumption-based truth maintenance represents alternative supporting environments and propagates consequences when assumptions change \citep{dekleer1986}. HaPPI maintains query-result probabilities under fact and probability updates through a provenance semiring \citep{gaur2021}. Evergreen connects semantic aggregate queries, confidence-sequence estimation and stopping, and provenance-based citations \citep{evergreen2026}. These methods establish the value of maintained support. Our contract attaches that support to the native statistical event and its funded coverage. Fact probability, justification and coverage allocation retain their different mathematical meanings.

\paragraph{Adaptive statistical inference.}
Prediction-powered inference rectifies model predictions with labeled data \citep{ppi2023}. Prediction-powered e-values allow predictable predictor and data-collection updates under bounds that keep the correction nonnegative \citep{ppe2025}. Prediction-powered risk monitoring uses predictable reliance on unlabeled predictions and analyzes intermittent labels and computation \citep{pprm2026}. These are reusable native inferential components. The present question concerns reassignment and reuse of accumulated guarantees while preserving the experiment that paid for them. Stochastic-prior joint moments \citep{rivasplata2020}, density-based prior control \citep{dziugaite2018} and positive anytime constructions \citep{morepac2024} provide valid selection-aware methods.

Recursive PAC-Bayes retains earlier risk confidence by recursively combining valid prior-risk and excess-loss bounds under a sequential sample partition \citep{recursive2024}. Alpha-spending pays for new tests; alpha-investing earns further testing wealth under an expected-false-discovery criterion \citep{foster2008}. Our lifetime guarantee uses spent failure allocations attached to native events. New consumers of an existing compatible set do not create new tests, and revocation does not restore spent allocation. Appendix~\ref{app:prior} also relates the contract to reusable holdouts, post-selection inference and e-value combination.

\begin{table}[t]\centering\small
\begin{tabularx}{\linewidth}{l X X}\toprule
Prior work & Core capability & Relationship to Adaptive Assurance\\\midrule
Recursive PAC-Bayes&Recursive prior-risk and excess-loss bounds&Retains confidence through qualified sequential experiments\\
ATMS / HaPPI&Contextual support; maintained query probabilities&Support and update machinery attached to native-event coverage\\
ENTRUST&Dynamic assurance cases and reconfiguration&Surrounding architecture for the event-ownership lifecycle\\
Evergreen&Semantic queries, confidence-sequence estimation and stopping, and provenance&Query machinery connected to editable native owners and dependencies\\
PPI / PPE&Prediction correction and predictable inference&Valid native components whose experiment identity is retained\\
PPRM&Risk monitoring with predictable reliance&Deployed-risk inference complements fixed-frame edit semantics\\
Stochastic-prior PAC-Bayes&Joint moment accounting for selected priors&Correctly prices the dependence exposed by the selected-start witness\\\bottomrule
\end{tabularx}
\caption{A knowledge map of the mature capabilities used by or adjacent to the framework. The comparison concerns mathematical objects and contracts, rather than runtime rankings of the original complete systems.}\label{tab:prior}\end{table}

Provenance semirings propagate symbolic source support through relational operations \citep{green2007}. These annotations provide dependency machinery; a native owner additionally fixes the inferential experiment, outcome binding and failure allocation. Confidence-set images, incomplete-information semantics \citep{imielinski1984}, exact proof checking, linear programming and finite-population inference \citep{wor2020} supply further machinery. The framework's connecting invariant is that native experiment identity and spent allocation remain fixed while edited consumers and dependencies change the justified image. The selected-start theorem constrains the statistical side of that invariant; the future-edit proposition constrains the retained information needed for future edits. Appendix~\ref{app:prior} records the relevant method assumptions and source versions.

\section{Discussion}
\subsection{Three consequences for adaptive systems}
\textbf{Statistical validity belongs to an experiment, rather than its stored numerical result.} A selected mixture can have valid future dynamics and unpaid starting wealth. A changed threshold can consume a funded confidence event directly, whereas a changed statistic requires a qualified relation or new evidence. Event ownership makes those operations explicit.

\textbf{Future editability depends on evidence origin.} A current feasible region can erase the difference between independent and derived support. Both the exact coordinate witness and the controlled real-record comparison show why later revocation needs retained dependencies. The $2^N/3^N$ comparison further shows that access requirements determine which distinctions a representation must preserve. Compression is therefore relative to the operations the system promises to support.

\textbf{Assurance recovery and competence adaptation are separate resource problems.} Reusing a trained model may require fresh assurance; a retained event may answer a new claim without changing the model. Selective invalidation can preserve observations and still leave acquisition planning expensive. The matched adverse results make this separation practical: fewer inspections, less search and smaller retained state are different benefits.

These observations suggest an assurance lifecycle alongside a system's competence lifecycle. A decision should expose its current claim, supporting event owners, active premises, removed dependencies and unresolved alternatives. The receipt derives that explanation from the computation itself. Source authentication and premise validation can then be attached as separate services with visible responsibility.

\subsection{Scope and limitations}
The reliability theorem assumes valid native events, truth and semantic/scope premises. Public recorded frames and simulated paid access test those assumptions' consequences; acquired clinical truth, physical digital-twin validity and autonomous future-drift assurance require additional evidence. The implemented relation class is linear and bounded. Certified outer relations can support a nonlinear consumer of an assured mean (Appendix~\ref{app:strengthening}); broader nonlinear consumers require their own outer constructions. Relaxations can widen refusal. Full JSON snapshots occupy about 259 MB for 291 heterogeneous certificates, so auditability has a material storage cost. Source acquisition, human reference construction and semantic authority remain charged or explicitly unknown. The selected-start result concerns the stated publisher version of the changing-prior extension and its unrestricted switch-count contract.

\section*{Broader impact}
The mechanism can expose unsupported assurance in adaptive monitoring. Its outputs can also be overtrusted if an implication certificate is mistaken for premise authentication. Deployment should expose the active premises and require their independent validation before consequential use. The health study analyzes recorded outcomes rather than clinical decisions; restricted pathology records were excluded.

\bibliographystyle{tmlr}\bibliography{references}
\appendix
\input{appendix/supplement}
\end{document}

%% file: sections/selected_start.tex
\section{A changing prior does not provide a free selected start}\label{a-changing-prior-does-not-provide-a-free-selected-start}

\subsection{The statement under examination}\label{the-statement-under-examination}

For each fixed \(\theta\), let \(W_t(\theta)=\exp P_t(\theta)\) be an adapted nonnegative supermartingale with \(\mathbb E W_1(\theta)\leq1\). The changing-prior extension in the current publisher version of \citet[Theorem 40, Equations 54--55]{chugg2023} permits switching to an \(\mathcal F_{t_i-1}\)-measurable prior \(\nu_i\) at time \(t_i\). It subtracts \(\mathrm{KL}(\rho\Vert\nu_i)\) from the unchanged accumulated \(P_t\). With \(s\) switches up to the current time, the displayed wealth boundary is \((s+1)(s+2)/\delta\), simultaneously over time and posteriors.

More explicitly, the source defines \(\mathfrak M_t(\theta)=P_t(\theta)-D_{\mathrm{KL}}(\rho\Vert\nu_i)\) between successive switches and asserts
\begin{equation}
\mathbb E_\rho[\mathfrak M_t(\theta)]\leq
\log\!\left(\frac{(t^\dagger+1)(t^\dagger+2)}{\delta}\right)
\quad\text{for all }t\geq1,\ \rho\in\mathcal M(\Theta),
\label{eq:source}
\end{equation}
with probability at least \(1-\delta\), where \(t^\dagger\) counts prior changes. We use its displayed accumulated \(P_t\); no reset or selected-start normalization is present in that statement.

Past measurability preserves the subsequent supermartingale property of a selected mixture. It does not give that mixture an expected start of at most one. The distinction is decisive even with one deterministic switch.

\subsection{Exact witness and unrestricted impossibility}\label{exact-witness-and-unrestricted-impossibility}

\begin{theorem}[Unpriced selected-start obstruction]
Under the printed assumptions above, the reported one-switch boundary is false. More generally, no finite boundary depending only on the number of switches and \(\delta\) controls all such unrestricted selected mixtures.
\end{theorem}

\textbf{Proof.} Let \(\Theta=\{1,\ldots,m\}\), let \(U\) be uniform on \(\Theta\), let \(\mathcal F_1\) be trivial, and let \(\mathcal F_t=\sigma(U)\) for \(t\geq2\). For \(0<\varepsilon<1\), define

\[
W_1(\theta)=1,\qquad
W_t(\theta)=
\begin{cases}
m-(m-1)\varepsilon,&\theta=U,\\
\varepsilon,&\theta\neq U,
\end{cases}\quad t\geq2.
\]

For every fixed \(\theta\), the conditional mean at time two is one. All later values are constant. Thus each process is strictly positive, adapted and a unit-start martingale. Begin with the uniform prior and make the single deterministic switch at time three to \(\nu_1=\delta_U\), which is \(\mathcal F_2\)-measurable. Take \(\rho=\delta_U\). The KL term is zero and the selected accumulated wealth equals \(m-(m-1)\varepsilon\) on every path.

With \(m=256\), \(\varepsilon=1/1000\), and \(\delta=1/20\), that wealth is \(255.745\), exceeding the claimed boundary \(6/\delta=120\) with probability one. For any finite proposed one-switch boundary, increasing \(m\) produces the same contradiction. \(\square\)

\begin{figure}[t]
\centering\includegraphics[width=.80\linewidth]{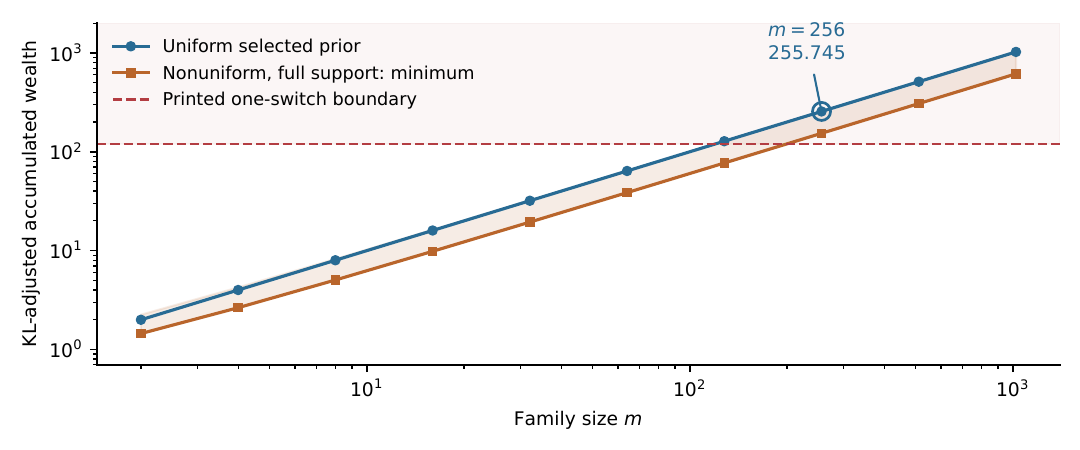}
\caption{Selected wealth grows with family size while the printed one-switch boundary stays fixed. The orange band spans the outcome-dependent full-support scores, and its lower curve gives their minimum. Shading above the boundary marks the violating region. The obstruction survives changes in prior support and outcome probability; the highlighted example uses $m=256$.}\label{fig:selected}
\end{figure}

The fixed-parameter optional-stopping step remains valid. The invalid inference transfers its unit expected start to a data-selected mixture. Here \(\mathbb E\int W_3\,d\nu_1=m-(m-1)\varepsilon\), although each fixed process has expectation one. A union bound over switch times cannot repair an unpriced start within a switch.

\subsection{Mature repairs and the precise consequence}\label{mature-repairs-and-the-precise-consequence}

The obstruction is unbounded even at one switch. A useful repair must price the selected start or impose additional structure on selection and accumulated wealth.

A density cap relative to a valid fixed prior pays for selection. Here the cap is \(m\), since the selected Dirac mass has density \(m\) relative to the uniform prior. \citet{dziugaite2018} provide density-based reasoning in data-dependent-prior analysis. \citet{rivasplata2020} retain the actual joint exponential moment for a stochastic prior kernel; the selected moment here is \(255.745\), rather than one. A third option discards selected past wealth and constructs a conditionally unit-start process from future increments.

These established remedies have different information and burden consequences. Fixed-prior PAC-Bayes and correctly priced data-dependent priors retain their original guarantees. The theorem isolates the accumulated-wealth changing-prior extension and rules out an unrestricted switch-count-only repair. The source versions and exact quantifiers appear in Appendix~\ref{app:proofs}.

\begin{theorem}[Corrected changing-prior bound with paid starts]\label{prop:paidstart}
Let $\Theta$ be finite, and let $W_t(\theta)>0$, $t\geq0$, be adapted supermartingales with $W_0(\theta)=1$. Let $0=T_0<T_1<T_2<\cdots$ be stopping times, allowing an uncreated suffix at infinity. At birth $T_i$, choose an $\mathcal F_{T_i}$-measurable prior $\nu_i$ and set
\[
 Q_{i,t}=\sum_{\theta}\nu_i(\theta)W_t(\theta),\quad
 S_i=Q_{i,T_i},\quad \alpha_i=\frac{\delta}{(i+1)(i+2)}.
\]
For $0<\delta<1$, with probability at least $1-\delta$, simultaneously for every created $i$, every $t\geq T_i$ and every posterior $\rho$,
\begin{equation}\label{eq:corrected}
 \mathbb E_\rho\log W_t(\theta)-D_{\mathrm{KL}}(\rho\Vert\nu_i)
 \leq\log\!\left(\frac{S_i(i+1)(i+2)}{\delta}\right).
\end{equation}
In particular, use prior $i$ from time $\tau_i=T_i+1$ until the next switch. Its start is observed at $\tau_i-1$; the switching decision is made before the next observation. For one birth and allocation $\alpha$, any deterministic $B\geq\mathbb E S$ instead gives the unconditional boundary $\log(B/\alpha)$.
\end{theorem}
Conditionally on $\mathcal F_{T_i}$, $Q_{i,t}/S_i$ has unit start. Ville's inequality and change of measure give the per-owner event; $\sum_{i\geq0}\alpha_i=\delta$ funds all switches. Appendix~\ref{app:proofs} proves the random-time statement. For an ordinary stopping switch that can react to its current observation, instead take the birth $T_i$ at the switch itself and pay $S_i=Q_{i,T_i}$; Equation~\ref{eq:corrected} already holds from that time. Using the earlier $Q_{i,\tau_i-1}$ requires $\tau_i-1$ to be a stopping time and a prior measurable there. Setting $W_0=1$ includes the initial phase. The original unit-expected-start family can be extended by this initial constant with trivial $\mathcal F_0$.

\begin{proposition}[Sharp finite-family starting-moment cost]\label{prop:selectioncost}
For $m=|\Theta|$ unit-start nonnegative supermartingales, an almost surely finite stopping time $T$ and any $\mathcal F_T$-measurable prior, $\mathbb E S\leq m$. The supremum of this expected selected start over strictly positive families is $m$.
\end{proposition}
Indeed, $S\leq\max_\theta W_T(\theta)\leq\sum_\theta W_T(\theta)$ and optional sampling bounds each expectation by one. The witness attains $m(1-\varepsilon)+\varepsilon\uparrow m$. Thus $m$ is the sharp universal \emph{starting-moment} price for unstructured finite-family selection. For its selected Dirac posterior, the density-cap bound and fixed uniform prior with $\log m$ KL cost have exactly the same normalized slack (Table~\ref{tab:remedies}); this identity does not compare arbitrary posterior contracts or establish an optimal tail boundary.

\paragraph{An ordinary Bernoulli example.}
Under the IID Bernoulli null $p_0=1/2$, take alternatives $q\in\{.3,.7\}$ and likelihood-ratio martingales $W_t(q)=(2q)^{s_t}[2(1-q)]^{t-s_t}$, where $s_t$ counts successes. At the deterministic birth $T=40$, a prespecified rule puts prior mass $.6$ on the more likely alternative and $.4$ on the other. Compare the paid mixture $E_{\rm paid}=Q_{T+n}/S$, parameterwise future rebasing with this same selected prior, and a fixed uniform prior over the entire prefix and suffix. Each experiment receives $\alpha=.05/6$; an endpoint value above $120$ rejects the null. Exact values on three realised count paths are
\input{tables/bernoulli_usefulness}
Here $a/b$ denotes successes/observations. The modest prefix gives $S\approx.544$ in the first two rows: retaining its relative alternative weights helps on an aligned suffix but hurts on an opposed one. The third prefix gives $S\approx16.127$, and its already-funded fixed-prior evidence is most useful. These count paths illustrate the stated likelihood model, rather than frequencies of a deployment advantage.

The information retained by the paid start has a precise interpretation. With $\widetilde\nu(q)=\nu(q)W_T(q)/S$, one has $Q_t/S=\sum_q\widetilde\nu(q)W_t(q)/W_T(q)$. Thus rebasing with the wealth-tilted prior reproduces the paid mixture exactly. The correction preserves relative past evidence; it does not supply an advantage over that appropriately initialized restart. Appendix~\ref{app:proofs} gives the KL form and the common-allocation comparison contract.

\citet{recursive2024} give a constructive neighbour: Recursive PAC-Bayes preserves earlier confidence through a recursive bounded-loss decomposition. Each new excess-loss term uses a prior independent of its evaluation suffix, and a separately valid earlier risk bound is carried forward. This supplies additional sequential structure rather than reassigning arbitrary accumulated martingale wealth to a selected prior.

The corrected theorem supplies a native event for the construction below: $E_i=\{\sup_{t\geq T_i}Q_{i,t}/S_i\leq1/\alpha_i\}$ has birth $\mathcal F_{T_i}$, allocation $\alpha_i$ and the selected prior/process binding. Its statistic is $f_i=\sup_{t\geq T_i}Q_{i,t}/S_i$, with native interval $[0,1/\alpha_i]$. It satisfies Equation~\ref{eq:funding} and enters Proposition~\ref{prop:reuse} directly. The owner composition applies to trajectory events; the later implementation specializes their constraint representation to bounded linear outcomes. Posterior or claim changes can consume this event; a fresh selected start creates a new owner.

The operational consequence is that a guarantee must remain attached to the experiment that paid for it. The selected-start witness changes the statistical owner while retaining accumulated evidence; the future-edit witness below changes the support through which that evidence reaches a claim. Adaptive Assurance preserves both kinds of identity: the native experiment's funding and bindings, and the dependencies that justify its current use.

%% file: tables/bernoulli_usefulness.tex
\[
(E_{\rm paid},E_{\rm rebase},E_{\rm fixed})=
\begin{cases}
(139.24,83.61,63.18),&24/40\text{ then }42/60,\\
(1.25,657.98,0.85),&24/40\text{ then }27/90,\\
(26.88,16.13,361.22),&28/40\text{ then }28/40,\\
\end{cases}
\]

%% file: tables/receipt.tex
\begin{table}[t]
\centering\small
\begin{tabularx}{\linewidth}{l X l l}\toprule
Edit & Evidence still valid & Range & Decision\\\midrule
Clipping revision & Old mean event and 32 relation-derived target points & [0.1866,0.8667] & OPEN\\
68 direct values added & 32 derived and 68 independent target points & [0.5956,0.5956] & REFUTE\\
Relation revoked & Old event remains; 68 independent target points & [0.4089,0.7290] & OPEN\\
Selective recovery & All 100 independent target points & [0.5956,0.5956] & REFUTE\\
Unqualified future & Recorded-frame evidence has no future binding & [0.0000,1.0000] & OPEN\\
\bottomrule\end{tabularx}
\caption{C-MAPSS claim: revised clipped mean $\leq.5$. Source allocation remains .02 throughout. Ranges are displayed outward to four decimal places; receipts retain exact rationals.}\label{tab:receipt}
\end{table}

%% file: tables/native_components.tex
\begin{table}[t]
\centering
\begin{minipage}[t]{.44\linewidth}
\centering\small
\setlength{\tabcolsep}{3pt}
\begin{tabular}{lrrr}\toprule
Study & Decisions & Qualified & Contrad.\\\midrule
Heart & 720 & 338 & 0\\
C-MAPSS & 60 & 33 & 0\\
Seeds & 36 & 18 & 0\\
Air & 36 & 16 & 0\\
\bottomrule\end{tabular}
\caption{\textbf{(a) Native lifecycle outcomes.} Three orders share each recorded frame. Contradictions are checked against recorded outcomes or permitted missing completions.}\label{tab:native}
\end{minipage}\hfill
\begin{minipage}[t]{.54\linewidth}
\centering\small
\setlength{\tabcolsep}{3pt}
\begin{tabular}{lrrrr}\toprule
Frame & \shortstack{Copy\\/5} & \shortstack{Stale\\points} & \shortstack{Unsup.\\/5} & \shortstack{Selective\\/ reset}\\\midrule
Heart & 1/5 & 303 & 5/5 & 202/303\\
C-MAPSS & 1/5 & 100 & 5/5 & 66/100\\
Seeds & 1/5 & 210 & 5/5 & 140/210\\
Air & 1/5 & 7,674 & 4/5 & 5,116/7,674\\
\bottomrule\end{tabular}
\caption{\textbf{(b) Component interventions.} Copy errors contradict outcomes; stale points lose relation support; unsupported decisions lack future scope. Recovery counts inspections restoring the same available information. Five-threshold grids are diagnostics.}\label{tab:ablation}
\end{minipage}
\end{table}

%% file: tables/lifecycle.tex
\begin{table}[t]
\centering\small
\begin{tabular}{llrrr}\toprule
Frame / retention & Acquisition & Qualified / queries & New exposures & Seconds\\\midrule
Adult / summary & endpoint guided & 60/72 & 10112 & 4.29\\
Adult / summary & uniform & 60/72 & 11598 & 4.08\\
Adult / summary & fresh & 64/72 & 3584 & 7.89\\
Adult / full & endpoint guided & 64/72 & 4352 & 4.52\\
Adult / full & uniform & 58/72 & 11563 & 5.30\\
Adult / full & fresh & 64/72 & 3712 & 5.80\\
Adult / no new access & endpoint guided & 26/72 & 0 & 0.04\\
Adult / no new access & uniform & 26/72 & 0 & 0.04\\
Adult / no new access & fresh & 26/72 & 0 & 0.04\\
BPI2012 / summary & endpoint guided & 64/72 & 4155 & 0.22\\
BPI2012 / summary & uniform & 64/72 & 4637 & 0.40\\
BPI2012 / summary & fresh & 64/72 & 2624 & 0.15\\
BPI2012 / full & endpoint guided & 64/72 & 2688 & 0.10\\
BPI2012 / full & uniform & 64/72 & 4176 & 0.13\\
BPI2012 / full & fresh & 64/72 & 2608 & 0.09\\
BPI2012 / no new access & endpoint guided & 7/72 & 0 & 0.03\\
BPI2012 / no new access & uniform & 7/72 & 0 & 0.03\\
BPI2012 / no new access & fresh & 7/72 & 0 & 0.03\\
\bottomrule\end{tabular}
\caption{Matched lifecycle: pooled qualified count over the configured orders; medians of target exposures and per-trajectory wall time. Exposure totals may revisit subjects. Source audits, training and computational-reference construction are additional costs. No resource exchange rate is assumed.}\label{tab:lifecycle}
\end{table}

%% file: appendix/supplement.tex
\section{Source contract and complete selected-start derivation}\label{app:proofs}
\subsection{What conditioning at a switch preserves}
The publisher statement is Theorem~40 of \citet{chugg2023}, with the accumulated process in Equations~54--55. Its current arXiv version 5 retains the relevant statement as Theorem~6.5. For every fixed parameter, $\exp P_t(\theta)$ is adapted and has a nonnegative supermartingale envelope with unit expected start. The changing prior is measurable with respect to the history immediately before its permitted switch. The posterior is arbitrary, and the statement is simultaneous over time and posterior. The KL term uses the current prior; the accumulated $P_t$ is not discarded at a switch.

\paragraph{Publisher statement, formulas and quantifiers.}
Theorem~40's hypotheses are $P(\theta)=(P_t(\theta))_{t\geq1}$, $\exp P(\theta)$ an $(\mathcal F_t)$-supermartingale, and $\mathbb E\exp P_1(\theta)\leq1$ for each fixed $\theta$. Section~6.5 permits stopping switch times $\mathcal T=\{t_1,t_2,\ldots\}$ and requires $\nu_i$ to be $\mathcal F_{t_i-1}$-measurable. Its displayed definition (54), transcribed with the source's process and index, is
\[
 \mathfrak M_t(\theta)=
 \begin{cases}
 P_t(\theta)-D_{\mathrm{KL}}(\rho\Vert\nu_0),&1\leq t<t_1,\\
 P_t(\theta)-D_{\mathrm{KL}}(\rho\Vert\nu_1),&t_1\leq t<t_2,\\
 \cdots&\cdots .
 \end{cases}
\]
The theorem states: ``with probability at least $1-\delta$, for all $t\geq1$ and $\rho\in\mathcal M(\Theta)$,'' followed by its displayed (55),
\[
 \mathbb E_\rho[\mathfrak M_t(\theta)]\leq
 \log\!\left(\frac{(t^\dagger+1)(t^\dagger+2)}{\delta}\right).
\]
Here $t^\dagger$ counts changes through time $t$. This quantifier wording and the formulas are reproduced from the publisher PDF, p.39; the surrounding assumptions are restated explicitly. The proof, continued on p.40, applies a unit-start maximal inequality to each selected mixture and then sums switch allocations. Neither definition nor proof resets $P_t$ or divides by its selected starting value. The inspected publisher PDF has SHA256 \texttt{4693d90691d77fb2905f41998885b8f7863966eb30cff9f3064fa8ffc258d5b8}; it is available at \url{https://jmlr.org/papers/volume24/23-0401/23-0401.pdf}. The January 2024 arXiv v5 has the same relevant equations as Theorem~6.5 (pp.36--37). The February 2023 v1 predates this changing-prior section, so it is not the version refuted. The counterexample addresses the accumulated-wealth extension \emph{as stated}; a future-only normalized experiment is a different valid contract.

Condition on the history before a switch. For each fixed parameter, subsequent supermartingale increments remain valid, and integrating with a prior measurable at that history preserves the future inequality. This proves a future supermartingale statement. It does not prove that the selected mixture's starting expectation is at most one. The source's fixed-prior mixing argument uses a fixed measure to exchange expectation and integration; a selected prior kernel requires their joint expectation. Ville's maximal inequality correctly prices this starting expectation. Applying a unit-start consequence at each selected switch requires the missing moment before any switch-index union bound is applied.

Our witness is finite, measurable, strictly positive and integrable, with a deterministic switch. It therefore avoids explanations based on zero capital, infinite support, optional-stopping regularity or numerical approximation. The positive fixed-prior theorem, change-of-measure inequality and correctly initialized Ville bound remain unaffected.

\subsection{Uniform witness and unrestricted quantifiers}
Let $U$ be uniform on $\{1,\ldots,m\}$, $\mathcal F_1$ trivial and $\mathcal F_t=\sigma(U)$ for $t\geq 2$. Define
\[
W_1(j)=1,\qquad W_t(j)=\varepsilon+m(1-\varepsilon)\mathbf 1\{U=j\},\quad t\geq 2.
\]
For each fixed $j$, $\mathbb E W_2(j)=\varepsilon+(1-\varepsilon)=1$ and later values are constant. Every $W_t(j)$ is a positive unit-start martingale. Put $P_t(j)=\log W_t(j)$. A uniform initial prior is independent of data. The switch at time 3 to $\nu_1=\delta_U$ is $\mathcal F_2$-measurable. Choose posterior $\rho=\delta_U$; its KL to $\nu_1$ is zero. On every outcome the selected wealth is $m-(m-1)\varepsilon$. At $m=256$, $\varepsilon=1/1000$, $\delta=1/20$, it equals $51149/200=255.745$, exceeding the printed one-switch $6/\delta=120$ boundary.

For any finite function $B(1,\delta)$ independent of family size, fix $\varepsilon\in(0,1)$ and choose an integer $m$ with $m(1-\varepsilon)+\varepsilon>B(1,\delta)$. This gives an admissible family violating that boundary on every outcome. The quantifier ranges over unrestricted finite families. It does not exclude a bound depending on cardinality, density, dependence, information or joint moment, nor a stronger class of initialization assumptions.

\subsection{Nonuniform and full-support adversarial extension}
For any probabilities $p_j>0$, $\sum_jp_j=1$, let $\Pr(U=j)=p_j$ and replace the second-time capital by
\[
 W_t(j)=\varepsilon+\frac{1-\varepsilon}{p_j}\mathbf 1\{U=j\},\quad t\geq 2.
\]
Each fixed column still has expectation one. Its selected start is $\varepsilon+(1-\varepsilon)/p_U$, with joint mean $\varepsilon+m(1-\varepsilon)$. For even $m$, choose alternating $p_j=3/(4 m)$ and $5/(4 m)$. With $m=256$, $\varepsilon=1/1000$ and $\delta=1/20$, the smallest selected wealth is 204.5962, exceeding 120 on every outcome.

The result also survives a full-support selected prior $\nu=(1-\eta)\delta_U+\eta p$, $\eta=1/4$. With $\rho=\delta_U$, exponentiating the PAC-Bayes score gives
\[
 \exp\{P_t(U)-D_{\mathrm{KL}}(\delta_U\Vert\nu)\}=W_t(U)\nu(U).
\]
In the displayed uniform and nonuniform families this exceeds 120 for every outcome. A restriction to full-support priors therefore does not by itself fix the missing selection cost.

\subsection{Valid alternatives and their original owners}
For a fixed prior $p$, the mixture $\int W_t\,dp$ remains a unit-start supermartingale. If a selected prior is pathwise dominated by $C p$ for deterministic $C$, its old-wealth mixture is at most $C\int W_t\,dp$. The ordinary maximal inequality then gives the appropriately enlarged boundary. In the uniform Dirac witness $C=m$. Density-based data-dependent prior methods \citep{dziugaite2018} are established positive machinery, not new results introduced here.

\citet{rivasplata2020} allow a stochastic prior kernel while retaining a joint exponential moment. For this witness that moment is $\varepsilon+m(1-\varepsilon)$, not one. The dependently selected prior is part of the experiment being priced. Theorem 16 of \citet{morepac2024} union-bounds events already valid at each time; its randomized-subsample Theorem 29 permits a prior seeing a supersample but not independent subsampling indices. Neither licenses the omitted initial moment in the printed extension.

Future-only rebasing is different from retaining accumulated selected wealth. For a deterministic switch $\tau$, strict positivity allows
\[
 R_t(j)=W_t(j)/W_{\tau-1}(j),\quad t\geq\tau-1.
\]
For each parameter, $R_{\tau-1}=1$ and $\mathbb E(R_t\mid\mathcal F_{t-1})\leq R_{t-1}$. A measurable prior chosen at $\mathcal F_{\tau-1}$ gives a conditionally unit-start mixture. New allocated maximal-inequality guarantees can therefore fund this future process. The same construction is valid at the pre-observation stopping births proved below, with conditional integrability supplied by the normalized nonnegative process.

The result is version-specific: publisher Theorem 40 and current preprint Theorem 6.5 retain accumulated wealth at a prior switch. The result concerns this accumulated-wealth statement and its unrestricted switch-count contract.

\subsection{Paid selected-start proof and remedy comparison}
For Theorem~\ref{prop:paidstart}, fix a created birth $T_i$ and condition on $\mathcal F_{T_i}$. Optional sampling for nonnegative supermartingales gives
\[
 \mathbb E[W_{T_i+n+1}(\theta)\mid\mathcal F_{T_i+n}]
 \leq W_{T_i+n}(\theta),\qquad n\geq0.
\]
The weights $\nu_i$ are measurable at birth, so $Q_{i,T_i+n}$ is a supermartingale on the shifted filtration. Strict positivity gives $S_i>0$. The normalized process $Q_{i,T_i+n}/S_i$ starts at one and is conditionally integrable. Conditional Ville bounds the probability of crossing $1/\alpha_i$ by $\alpha_i$. Finite change of measure gives
\[
 \mathbb E_\rho\log W_t-D_{\mathrm{KL}}(\rho\Vert\nu_i)\leq\log Q_{i,t}
\]
for all posteriors simultaneously; infinite KL gives a vacuous score. Consequently $\Pr(E_i^c\mid\mathcal F_{T_i})\leq\alpha_i$ for $E_i=\{\sup_{t\geq T_i}Q_{i,t}\leq S_i/\alpha_i\}$. For a birth that is never created set its failure indicator to zero; the conditional argument applies on each finite-birth event. The countable union bound, or Proposition~\ref{prop:reuse}, now uses $\sum_i\alpha_i=\delta$. At each active phase the index is the number of prior changes, reproducing the switch-count factor with the previously missing $S_i$.

No future selection enters this conditional argument: $T_i=\tau_i-1$ is a stopping time, the prior is measurable there, and use begins at $\tau_i$. Ordinary stopping of $\tau_i$ alone does not guarantee this earlier-birth property. If the switching decision observes time $\tau_i$ first, use birth $T_i=\tau_i$ and the observed $Q_{i,\tau_i}$; its future event cannot retroactively fund earlier observations. For an almost surely finite birth, $S_i\leq\sum_\theta W_{T_i}(\theta)$ is integrable with expectation at most $m$. The unconditional version instead applies Ville with deterministic threshold $B/\alpha$ when $B\geq\mathbb E S$.

For Proposition~\ref{prop:selectioncost}, nonnegative optional sampling yields $\mathbb E W_T(\theta)\leq W_0(\theta)=1$, even for an unbounded almost surely finite $T$. Summing proves the upper bound. The positive witness at deterministic $T=2$ and $\nu=\delta_U$ approaches $m$ as $\varepsilon\downarrow0$, proving sharpness of the expected-start constant. In the Dirac witness the two safe comparisons are algebraically identical: $W_t(U)/m\leq1/\alpha$ under a fixed uniform prior and $W_t(U)\leq m/\alpha$ under its density cap. Their score-to-boundary ratios agree; general posteriors and tail-optimality are separate questions.

\begin{table}[t]\centering\small
\begin{tabularx}{\linewidth}{l r r X}\toprule
Construction & Witness score & Boundary & Paid quantity / retained information\\\midrule
Printed one-switch extension &255.745&120&Switch count; accumulated selected wealth\\
Fixed uniform prior &.999004&20&Original prior; score includes $\log m$ KL cost\\
Density cap &255.745&5120&$C=m=256$; accumulated wealth retained\\
Joint starting moment &255.745&5114.9&$B=255.745$; selected dependence retained\\
Observed-start accounting &255.745&5114.9&$S=255.745$; accumulated score explicitly priced\\
Future-only rebasing &1&20&Unit future start; old wealth removed from new score\\
Recursive PAC-Bayes &\multicolumn{2}{c}{Different risk contract}&Earlier risk bound plus independent-suffix excess-loss bounds\\\bottomrule
\end{tabularx}
\caption{Constructive remedies on the uniform positive witness, $m=256$, $\varepsilon=.001$, $\delta=.05$. Original-process scores are $\exp\{\mathbb E_\rho\log W_t-D_{\mathrm{KL}}(\rho\Vert\nu)\}$ for $\rho=\delta_U$; the future-only row instead uses the rebased process. Valid alternatives allocate $\alpha=.05$ to one experiment; the printed comparison uses its $6/\delta$ switch boundary. Recursive PAC-Bayes has the distinct IID bounded-loss and sample-partition contract of Theorem~5 \citep{recursive2024}. The rows show what each contract pays for, rather than rank bounds under a common experiment.}\label{tab:remedies}\end{table}

Recursive PAC-Bayes carries $B_t=E_t+\gamma_tB_{t-1}$ through disjoint training increments. The prior for each excess-loss term and $\gamma_t$ are chosen before that term's evaluation suffix; earlier confidence retains its own valid bound. This constructive contract and the paid-start theorem express the same ownership requirement through different inferential objects.

\paragraph{Bernoulli comparison and tilted rebasing.}
The three numerical paths in Section~\ref{a-changing-prior-does-not-provide-a-free-selected-start} use a known IID Bernoulli null $p_0=1/2$ and positive likelihood ratios for $q=.3,.7$. At each step the null conditional mean of the likelihood increment is one. Equivalently, for $\widehat p_t=s_t/t$ and $t>0$,
\[
 \log W_t(q)=t\{\operatorname{kl}(\widehat p_t\Vert1/2)
                    -\operatorname{kl}(\widehat p_t\Vert q)\}.
\]
The fixed prior $p(q)=1/2$ is born before the prefix. At $T=40$ the selected rule chooses $\nu$ from the prefix likelihoods, breaking a tie towards $.7$. Its paid experiment and its same-prior rebased counterpart are born at that history. With suffix ratios $R_{T+n}(q)=W_{T+n}(q)/W_T(q)$, their endpoint values are respectively
\[
 E_{\rm paid}=\frac{\sum_q\nu(q)W_T(q)R_{T+n}(q)}{S},\qquad
 E_{\rm rebase}=\sum_q\nu(q)R_{T+n}(q),\qquad
 E_{\rm fixed}=\sum_qp(q)W_T(q)R_{T+n}(q).
\]
All receive the same per-experiment allocation $1/120$. They are alternative native experiments, not a deployment rule selecting the most favorable endpoint without further allocation. All counts, likelihoods and prior rules are available to every method; the difference is which birth and retained weighting fund the guarantee. The three prefix/suffix pairs are fixed before computation, and exact rational arithmetic verifies the displayed scores. They do not estimate average power or first-crossing time.

For any positive finite family, $\widetilde\nu(\theta)=\nu(\theta)W_T(\theta)/S$ is a birth-measurable prior with the same support as $\nu$. For every posterior of finite KL,
\[
 \mathbb E_\rho\log W_t-D_{\mathrm{KL}}(\rho\Vert\nu)-\log S
 =\mathbb E_\rho\log\frac{W_t}{W_T}
       -D_{\mathrm{KL}}(\rho\Vert\widetilde\nu).
\]
This is an algebraic translation between the paid-start score and future rebasing with the tilted prior. With a selected Dirac prior the two coincide even without tilting. A numerical gain relative to rebasing with unchanged $\nu$ therefore concerns retained between-alternative weights, not free use of past likelihood or a stronger maximal inequality.

\section{Conditional projection and exact verification}\label{app:projection}
\subsection{Countable event creation}
Enumerate possible owner births by $j$. If an owner is not created set its allocation and failure indicator to zero. Taking conditional expectation at its birth history gives $\Pr(E_j^c)\leq\mathbb E\alpha_j$. Monotone convergence and a union bound yield
\[
 \Pr(\exists j:E_j^c)\leq\sum_j\mathbb E\alpha_j
 =\mathbb E\sum_j\alpha_j\leq\delta.
\]
On the complementary event every active native row contains the bound truth. Truthful points and correct qualified relations also contain it. Deleting a constraint only enlarges the compatible set. A data-dependent consumer has its true value in its outer image, so the threshold decision is correct. Budget paths belong to one trajectory, not the pooled collection of counterfactual methods and domains.

The assertion requires more than an owner identifier. An event must actually satisfy its stated conditional guarantee for the original statistic and outcome binding. A string naming a birth time does not prove sampling honesty. Wrong external truth, unqualified scope changes and forged provenance are excluded by the premises, not solved by this theorem.

\subsection{Rational weak duality with residuals}
For $Az\leq b$ and $\ell\leq z\leq u$, any $\lambda\geq 0$ gives
\[
 d^\top z=\lambda^\top Az+(d-A^\top\lambda)^\top z
 \leq\lambda^\top b+\sum_r\max\{r_r\ell_r,r_ru_r\}.
\]
Numerical dual proposals are converted to rationals with denominator at most $10^8$. The residual is retained in full; it is not rounded away or presumed zero. Applying the inequality to both signs gives an outer range even for nonoptimal multipliers. Zero multipliers give the box range. The replay rejects inactive-row references, negative multipliers, mismatched native statistics/bindings, inconsistent endpoints and unfunded owners. It verifies represented consequences, not provenance authenticity.

For integer binary sums, outward integer rounding preserves validity. Generic continuous projection may be wider than exact categorical or integer projection. That can reduce coverage, but does not justify unsafe narrowing. Infeasible numerical reports alone do not produce a contradiction certificate.

\subsection{Binary compression and edit limits}
Partition unknown subjects by their tuple of coefficients in all active native events and the current claim. Replacing a cell by its count preserves every linear statistic. Every subject assignment induces a feasible cell-count assignment, and every integer count in $[0,n_{cell}]$ is realized by that many ones. Thus feasible objective sets coincide. Access, relation or point-provenance restrictions distinguishing subjects require refinement.

This current-range sufficiency is not a future-edit sufficiency theorem. Grouping must be refined when a future operation distinguishes origin, availability, bindings or premises. The following construction makes the preserved language explicit.

\subsection{Future-response sufficiency and a restricted memory bound}\label{app:editlanguage}
Let $q(h)$ be the required observable response of a retained history, and let $u$ range over a declared operation language. Use a common operation alphabet; an unsupported operation leaves the state unchanged and emits a distinguished refusal response. Define
\[
 h_1\sim_{\mathcal U}h_2
 \quad\Longleftrightarrow\quad
 q(u h_1)=q(u h_2)\quad\text{for every finite declared operation sequence }u.
\]
A deterministic representation from which those responses can be recovered must distinguish inequivalent histories. Equivalence is stable under prefixing a further operation, so the quotient supports a well-defined transition and response. This is a standard response-congruence construction. It asserts neither a computable minimal quotient in an unrestricted language nor low construction cost. Stochastic access additionally requires preservation of its response law and cost, not just a present deterministic answer.

For the two-coordinate witness in Proposition~\ref{prop:edit}, both original feasible sets are the singleton $(0,0)$. After revoking $T_x$, the first feasible set is $\{(a,a):a\in[0,1]\}$ and the second remains the singleton. Their target intervals are $[0,1]$ and $[0,0]$. This distinguishes evidence origin even when no present-outcome information has been compressed.

For a restricted quantitative example, take one original fixed inventory of $N$ truthful IID Bernoulli outcomes, bound to its original scope and funded at $\alpha\in(0,1)$. The one-sided fixed-$N$ Clopper--Pearson upper endpoint $U_N(k;\alpha)$ is strictly increasing in $k=0,\ldots,N$, with $U_N(N;\alpha)=1$ \citep{clopper1934}. For $k<N$, the endpoint solves $\Pr_p\{\mathrm{Bin}(N,p)\leq k\}=\alpha$. The tail is strictly decreasing in interior $p$ and strictly increasing in $k$, so its unique root increases with $k$; the final endpoint is one. Each retained cell is $*,0,1$. After arbitrary outcome-dependent erasure, conservative completion uses
\[
 U_N(k_1+k_*;\alpha),
\]
where $k_1$ counts retained ones and $k_*$ unknown cells. On the original full-inventory coverage event, the true full count never exceeds $k_1+k_*$; hence every such completion remains sound. Erasure does not manufacture a new IID sample, a new event or a refunded allocation.

For all upper-threshold outputs and every named erasure, the bit $b_i=\mathbf 1\{i\text{ is retained zero}\}$ is sufficient: the upper-count index is $N-\sum_i b_i$ and erasing $i$ sets $b_i=0$. To separate two different bit vectors, erase every other coordinate. Their indices are $N-1$ and $N$, separated by a threshold between the strictly ordered endpoints. Thus exactly $2^N$ response classes, or $N$ information bits, are necessary and sufficient for this language. Unknown and retained-one cells are genuinely equivalent here.

Now add an honest response reporting whether each original item is retained or requires paid reacquisition. Unknown and retained-one cells become distinct; retained-zero and retained-one cells were already distinct. Every ternary configuration is therefore separated, giving $3^N$ classes and $N\log_2 3$ information bits. A direct two-bit-per-cell encoding is a simple upper representation. These counts concern an expanded output/access contract, not a universal lower bound for all assurance tasks. Inventory, scope, protocol, funding, heap storage and source acquisition are additional burdens.

The finite calibration enumerates all 363 ternary states for $N=1,\ldots,5$, every named-erasure subset and 33,033 separating pairs. It independently recovers the $2^N$ and $3^N$ partitions. The calibration also compares 60 actual beta-quantile endpoints with 45 independent binomial-tail inversions and verifies strict endpoint order. Enumeration checks arithmetic and response semantics; sampling risk is governed by the native-event assumptions.

\subsection{Conditional replay versus truth authentication}
A native numeric row asserts $z\in[2/5,3/5]$ for the claim $z\leq3/10$, which is refuted. A coherent forged document keeps the same asserted statistic, bindings, experiment, premises and allocation, but replaces the interval by $[1/10,1/5]$ and recomputes the dual witness and hash. Replay admits the claim. For diagnostic external truth $z=1/2$ that admission is false, because the forged native event is not valid. Replay remains a correct implication from its asserted input; it has not verified stochastic interval generation or document authenticity. Narrowing only a reported range without changing its premises is a different manipulation, and is rejected. The implementation tests both manipulations.

\section{Experimental protocols and authority passports}\label{app:experiments}
\subsection{Source reconstruction and common access}
The release downloads raw public inputs by hash. Adult preprocessing standardizes columns 0, 2, 4, 10, 11 and 12 and one-hot encodes all others, retaining the missing category. Logistic regression uses liblinear, 500 iterations and seed 0, with regularization $C=1$ or $.01$. A source-only tree has depth 3. The source model's $.55$ threshold is a declared replacement. Test labels are not used to fit these models. Their evaluator availability is not a claim that real labels are cheap or scarce.

BPI2012 contains 13,087 traces and 262,200 events. XML event order is retained without timestamp sorting or lifecycle filtering. The first 70\% yields 9,160 traces and 3,149 variants. The provided reference accepts 5,826 traces through independently checked firing paths; the state equation and its left-nullspace necessity refute 3,067 \citep{murata1989}. A further 267 remain unresolved before cost-based alignment of four fixed source traces \citep{adriansyah2011}, using PM4Py 2.7.20 \citep{pm4py2023}. These checks qualify four additional negative traces, leaving 263 unknown and 8,897 known. The reference file is supplied computational authority; its historical human construction is unknown. These outcomes do not certify all original traces or a clinical workflow. The release recomputes the paths, exact basis and mature alignment checks rather than reading a hidden label cache.

The binary lifecycle seeds are 20261004120--20261004127. Each of 24 domain--retention--method cells contains eight nine-step trajectories. Source audit size is 2,048 (Adult) or 512 (BPI). Claims use declared thresholds .18/.20 and .4 as defined by the public program, a subgroup, replacements, a common source-selected target, point withdrawal, native-premise revocation and an unsupported version. Target roots allocate .003 and use looks at 32, 128, 512 and 2,048, capped at available population size. Endpoint, uniform and fresh rules share targets and budgets. Summary, full, 128-point and no-new-access retention change the information contract, not the truth owner.

\subsection{Native ordinal and numeric protocols}
The heterogeneous experiment uses seeds 20261004500--20261004502. The three Heart Disease frames have 303, 123 and 200 records with coded outcomes 0--4; the claim is their mean divided by 4. A 64-record source presence audit allocates .02. Transport and fresh strategies receive the same source work and retained points. A target mean owner allocates .02 across at most four looks, at sizes 32, 64, 128 and $N$, each at .005. Independently observed target categorical points retain only their target-truth premise. Coarse-relation revocation and target-truth revocation therefore have different consequences.

C-MAPSS FD001 and FD002 supply 100 and 259 reference RUL outcomes. Source values clipped at 200/normalized by 200 and target values clipped at 125/normalized by 125 are declared before audits. Thirty-two source observations and their .02 event use a qualified clipping envelope. Derived target points explicitly depend on both source truth and the clipping relation. Later independent target points do not. Relation revocation leaves the latter intact; 32 direct target observations restore the former subjects. A simulator future-version assumption is not supplied and is refused. This is output-contract revision on a simulated reference, not physical version-drift acquisition.

Seeds provides 210 kernel measurements. The recorded area is clipped to $[0,30]$ and normalized; the range defines the estimand, not a bound estimated from target data. Sixty-four source measurements fund .02. Unit conversion changes the consumer; an unrelated perimeter channel remains unconstrained; independently observed target area survives source-area revocation. The additional DryBean source was investigated but not retained because it did not expose a distinct operation.

Air Quality uses 9,357 dated records, 7,674 available CO values and 1,683 missing values. The native numeric estimand clips CO to $[0,10]$ and normalizes; negative missing sentinels are excluded from observed sampling and remain unconstrained full-frame coordinates. The source uniform sample has 256 available rows and allocation .02. An available census never acquires missing truth. The separate binary control uses the declared recorded exceedance CO $>2$ and leaves all missing binary outcomes arbitrary.

\subsection{Native statistical components}
Binary fixed-look calibration inverts exact hypergeometric tails under ideal uniform sampling without replacement. Multiple looks are union-funded before observing their outcomes. Numeric calibration uses the established bounded Hoeffding inequality for a fixed without-replacement sample, with radius $\sqrt{\log(2/\alpha)/(2 n)}$, clipped to the declared $[0,1]$ mean range. A complete census is exact. This intentionally conservative component does not reproduce the sharper time-uniform without-replacement construction of \citet{wor2020}.

The implementation computes a rational upper enclosure of the logarithm and square root before using the radius. Retained sample points and aggregate events are distinct objects. Data-driven uniformity, semantic truth, range and sampling implementation are not proved by these arithmetic checks. Temporal dependence in a fixed recorded frame does not invalidate ideal design-based random sampling.

\subsection{Fair comparators and privileged references}
Matched endpoint, uniform and fresh audit rules share the claim, public predictions, source event, retained paid points, target truth channel and allocations. The generic linear/integer control shares constraints and independently validates projection. Census gets additional outcomes and is charged. Exact finite enumeration gets complete small compatible-world access and is a computational reference. Invalid copying diagnostics deliberately omit a component and are not legitimate performance baselines. No method is declared inefficient solely because it loses on one metric.

The public health and sensor binary relation control uses seeds 20261004200--20261004203, three thresholds .25,.5,.75, two retention arms and two strategies. It produces 432 multivalued-file lifecycle decisions, 12 unavailable-Hungarian-outcome requests and 216 sensor-reference decisions. Its 43 receipts expose relation lineage and missing-access refusals. The heterogeneous suite's 852 decisions instead preserve native ordinal/numeric functionals; the studies do not share an arbitrary binary target merely to match implementation.

\subsection{Units, dependence and negative evidence}
Every reported decision is a query on a fixed recorded frame. Sampling orders, thresholds, retention arms and lifecycle states are correlated. Neither 1,728 binary nor 852 heterogeneous decisions are independent rare-event trials. Zero observed contradiction is not a statistical estimate of zero risk. The lifetime guarantee follows from the native experiment assumptions and funding theorem.

The release retains exact finite omissions for predicate identity, semantic relation, version scope, point-premise provenance and refusal; it tests noninvertible categorical/numeric aliases. Unit tests additionally reject narrowed certificates after rehashing, negative multipliers, changed bindings and duplicate funding. Replay rejection does not certify external authentication.

Source metadata cannot supply unavailable outcomes. Exact projection can be costly, and a correct stop rule need not select an optimal audit.

\section{Prior-method contracts}\label{app:prior}
ENTRUST combines verifiable controller models, partially instantiated stochastic system/environment models, reusable assurance arguments, a verified controller platform, runtime monitoring and reconfiguration. Its design-time stages establish controller/platform obligations; its runtime stages obtain adaptation evidence and update assurance cases. These capabilities are established in its accepted manuscript and two-domain evaluation. An event-owned statistical operator can be embedded in that surrounding assurance architecture.

Evergreen version 2 represents claims over a fixed finite relation using quantified semantic queries. Verification-aware stopping and confidence-sequence estimation are combined with query optimization and semiring provenance. This connects semantic evidence to a verified aggregate. Our edit contract maintains the statistical experiment behind that evidence as consumers and premise dependencies change. The comparison concerns these mathematical contracts; empirical runtimes of the original complete engines are separate from the component controls used here.

PPE's final ICML 2025 Theorem 2.1, proof and Algorithm 1 establish a bounded nonnegative random-query correction with predictable model and data-collection updates. PPRM's final ICML 2026 Theorem 3.2 controls false alarms with predictable reliance on unlabeled predictions. Its intermittent-label extension additionally assumes bounds for unlabeled periods; its computational analysis is separate from that guarantee. PPI combines prediction-based estimation with labeled correction. These methods give positive native inference components. The fixed-frame consumer/edit lifecycle preserves their applicable event rather than treating a change in statistical ownership as a predictable future update.

Rivasplata's stochastic-prior change of measure retains the joint exponential moment; Dziugaite--Roy's density control prices prior dependence. Valid future-only rebasing supplies a third response to the selected-start obstruction. Confidence-set projection, incomplete-information semantics, linear optimization and finite-population concentration support the surrounding construction.

Other selection-aware methods control different interfaces. Reusable holdouts mediate adaptive access through stability under an IID query-and-budget contract \citep{dwork2015}. PoSI provides simultaneous inference over linear-submodel parameters \citep{berk2013}. E-value combination merges already valid inputs under specified rules, such as fixed averaging \citep{vovk2021}. Their access, model and merging conditions belong to the native inference service.

ATMS Section 4 represents minimal consistent supporting environments, propagates alternative justifications and manages contradictions and subsumption. This supplies assumption-conditioned support maintenance. HaPPI Sections 2 and 4--5 use independent fact probabilities, a provenance semiring and maintained expressions under insertion, deletion and probability change. Query-result probability describes existence in a possible-world model. A funded coverage event instead describes the failure probability of an inferential procedure. The framework preserves both contextual support and the identity/allocation of the native coverage event.

\section{Typed burden, storage and reproducibility}\label{app:burden}
Source acquisition, training, reference construction, target observations, repeated exposures, optimization, arithmetic replay and retained memory are separate quantities. The supplied reference and semantic relations are authority, not free consumed compute. Unknown human preparation remains unknown. The current compact binary grouping does not imply economical full provenance.

\begin{figure}[t]\centering\includegraphics[width=\linewidth]{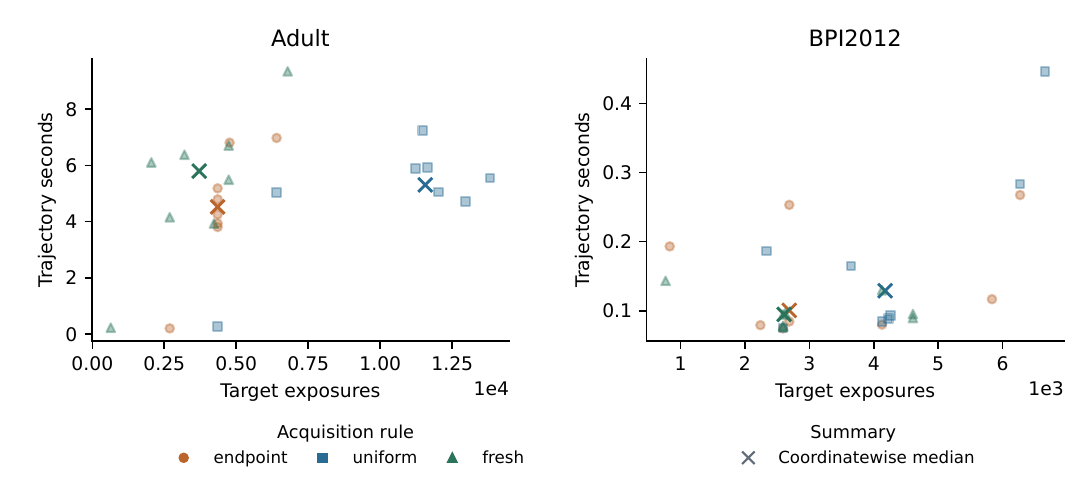}
\caption{Full-retention resource pairs across eight common orders. Individual trajectories share marker size and opacity; shape and color identify the acquisition rule. Same-color crosses mark coordinatewise medians. Exposures and time remain separate resource dimensions.}\label{fig:burden}\end{figure}

The native suite exports 291 full JSON certificates, about 259 MB with source bindings and constraints duplicated across snapshots. Its literal storage is a measured negative limitation. A full receipt can still offer auditability when it is computationally expensive. The code supplement excludes raw NASA redistribution and generated large outputs; downloaded version-pinned public inputs and commands reconstruct them. Timings are local measurements and may change with dependency versions, hardware and concurrent work. Claims of theory validity do not depend on wall-clock.

The artifact contains exact future-edit and conditional-forgery checks and 43 focused tests. Its code is MIT-licensed; source data and the official template retain separate terms. Reproduction materials have been prepared and will be released separately. The implementation provides theorem checks, raw source preparation, lifecycles, native-outcome studies and table/figure generation. New executions record code hashes, protocol and failures without overwriting outputs. Predictions and workflow outcomes are rebuilt from hash-pinned public inputs. No private data, internal history or hidden cached frame is required.

\section{Controlled real-frame component interventions}\label{app:components}
The component study uses four fixed frames: Cleveland recorded codes, C-MAPSS FD001, Seeds and Air Quality. It is a targeted functional intervention, rather than a sample of naturally occurring guarantee failures. Source outcomes are completely inspected; every third available item is additionally inspected on the target channel. Heart maps binary presence to the ordinal range $[x/4,x]$. C-MAPSS uses the declared clipping relation. Seeds uses an identity after normalization of its unit-converted area; its wrong-functional diagnostic copies that mean to the unrelated perimeter normalized by 20. Air relates observed reference values to the same coordinates in the expanded frame and leaves missing outcomes free.

The source census is charged and owns a fixed statistic with allocation .02. Its deterministic truth is a strong experimental control. Copying that statistic to target coordinates under the same owner is rejected; the omission diagnostic ignores identity and treats the source interval as a target interval. All domains use the predeclared thresholds $.1,.25,.5,.75,.9$. Each has one contradictory copying decision among these five requests. Conclusions are checked against recorded truth, with arbitrary missing completions permitted in Air. This diagnostic removes a required contract, and is distinct from a fair acquisition competitor.

For dependency withdrawal, we remove the relation and compare the active rows with those before the edit. Each available item has one derived target point represented by two inequalities. The stale-point counts in Table~\ref{tab:ablation} count points rather than rows; semantic relation rows are counted separately in the artifact. The lawful system retains the independent every-third target points. Selective recovery purchases each remaining available target once. Full reset purchases every available target again. The comparison fixes the restored information, while construction, projection, replay and storage are recorded separately. Air retains 1,683 missing outcomes in both recovered regions.

A future-scope request requires an unqualified premise. The complete mechanism returns OPEN. Removing that requirement, or substituting the last qualified answer, yields unsupported decisions. This is a scope error, even if the recorded-version answer coincidentally remains correct. Statistical selection accounting retains its separate exact witness; these real-frame omissions do not estimate a selection-error rate.

The two C-MAPSS provenance histories use all 100 old outputs in both arms. Relation-derived target points in one arm and 100 independent target inspections in the other fix exactly the same 200 coordinates. Their current joint regions are identical singletons, not merely equal projected means. The same relation edit then removes all target constraints in the first history but leaves direct target points in the second. The artifact exports and replays all four resulting receipts. This controlled construction uses actual recorded values without asserting that the paired histories occurred in a field deployment.

\section{Supplementary statistical recovery}\label{app:strengthening}
\subsection{Partial audits and calendar-defined consumers}
Forty prespecified source orders, 20261007100--20261007139, are paired across retention arms. C-MAPSS FD002 uses the cap200/cap125 relation and 32 source observations. Air Quality uses a 256-observation source audit of the complete available frame, followed by requests for March 2004, September 2004 and March 2005. Date metadata defines these consumers; their available reference frames contain 492, 556 and 724 observations. These are calendar-defined task changes, rather than evidence of physical drift. Missing analyzer outcomes remain outside each explicitly available-frame mean.

Selective retention keeps qualified target points; reset discards them. Both pay the same source work and receive the same target priority permutation, restricted to their remaining populations. The original full-frame mean is not renamed as a monthly mean. Instead, the target mean combines the retained exact sum with a uniform audit of the fixed remaining population. Before drawing, a new owner allocates .02 across at most five fixed looks, at 32,64,128,256 and census, each receiving .004. Source allocation .02 remains spent. The strong ordinary-inference comparator adds the retained sum to the scaled complement interval and intersects \emph{all} funded looks. It matches the framework on every query. Keeping only the latest interval can be looser, which is an information-retention difference rather than a new statistical bound.

\input{tables/supplementary_recovery}

The order intervals expose the different roles of retention. At census it can avoid reacquiring already observed target values. Before census the paired median difference can include zero. In this common-priority panel retaining points never increases the observation count, but this is not a uniform acquisition theorem. Independently drawn arm orders can reverse a comparison: the preserved three-order panel includes a C-MAPSS $.75$ case requiring 227 new values under retention and 128 under reset. The paired study controls that sample-path difference and reports months separately, rather than treating one calendar order as evidence of a retention effect.

Every projected range in the paired panel contains its recorded reference, and every qualified decision agrees with it. Orders are randomized units conditional on fixed recorded frames; thresholds, looks and lifecycle queries within each order remain dependent. The marginal 95\% median intervals invert a binomial sign test: for sorted costs $X_{(1)},\ldots,X_{(n)}$, use $[X_{(k)},X_{(n-k+1)}]$, where $\Pr\{\operatorname{Bin}(n,1/2)\leq k-1\}\leq.025$ and $k$ is largest. For $n=40$, $k=14$; ties make coverage conservative. These intervals measure order variation, not deployment risk. Source observations, new exposures, inference, projection, verification and wall time remain separately recorded.

\subsection{Practical cache reuse and a tagged control}
Two common reuse rules either inherit a source confidence interval unchanged or relabel source samples without their origin. We test them on four declared edits: Heart presence-to-grade, C-MAPSS FD001 clipping, Seeds area-to-perimeter, and Air available-to-full-frame. The orders, source observations and public semantics are matched. A tagged completion control keeps channel identity and applies qualified transformations, leaving unrelated outcomes unknown. On this subcase its ordinary annotations and projection reproduce the framework's valid edit responses.

\input{tables/cache_controls}

The fixed threshold grid is $K=j/400$, $j=0,\ldots,400$, declared before outcomes. Source-census reuse exposes the full difference between the two functionals; repeating its order cannot increase statistical evidence. The release includes partial 32- and 128-observation endpoints, where source confidence width can conceal a semantic mismatch. Consequently a quiet partial audit does not justify copying a guarantee. The tagged control also shows the framework's proper boundary: provenance maintenance and a union bound can implement these operations when they preserve the native statistic, allocation and bindings. Event ownership makes that contract explicit and checkable; it supplies lifecycle semantics and receipts rather than a stronger confidence interval. Air's permissible missing completions measure unsupported assurance, not an observed analyzer error rate.

\paragraph{Nonlinear-consumer remark.}
An assured nonnegative mean $m\in[a,b]$ can support $q=m^2$ through tangents $q\geq2tm-t^2$ and the chord $q\leq(a+b)m-ab$, whose dependencies include the certified mean-range premises. The executable 128-of-259 C-MAPSS pilot checks this outer construction and its revocation. It illustrates compilation into the linear layer, not a claim about the mean of squared outcomes or a general nonlinear solver; its detailed receipts remain in the release.

\subsection{Spending belongs to owner births}
An allocation illustration creates four future owners with $\alpha_j=.05/2^{j+1}$, $j=0,\ldots,3$. Twenty threshold consumers per owner reuse the current region without further spending. Revoking each owner premise removes its rows while retaining its expenditure; total allocation is $3/64$. This is alpha-spending attached to experiments. Alpha-investing supplies a different sequential expected-false-discovery criterion \citep{foster2008}; its earned testing wealth does not replace the pathwise lifetime-error contract here. The illustration checks the bookkeeping distinction, not a new power comparison.

%% file: tables/supplementary_recovery.tex
\begin{table}[htbp]\centering\small
\begin{tabular}{llrrrr}\toprule
Frame & $K$ & Selective & Reset & Paired difference & \shortstack{Before census\\selective / reset}\\\midrule
C-MAPSS & 0.5 & 227 [227,227] & 259 [259,259] & -32 [-32,-32] & 2/40; 0/40\\
C-MAPSS & 0.75 & 128 [128,128] & 128 [128,256] & -14.5 [-128,0] & 38/40; 40/40\\
C-MAPSS & 0.9 & 32 [32,32] & 32 [32,64] & 0 [0,0] & 40/40; 40/40\\
Air Mar04 & 0.25 & 476 [475,477] & 492 [492,492] & -16 [-17,-15] & 0/40; 0/40\\
Air Mar04 & 0.5 & 64 [64,64] & 64 [64,64] & 0 [0,0] & 40/40; 40/40\\
Air Sep04 & 0.25 & 538 [537,539] & 556 [556,556] & -18 [-19,-17] & 0/40; 0/40\\
Air Sep04 & 0.5 & 64 [64,64] & 64 [64,64] & 0 [0,0] & 40/40; 40/40\\
Air Mar05 & 0.25 & 700 [698,701] & 724 [724,724] & -24 [-26,-23] & 0/40; 0/40\\
Air Mar05 & 0.5 & 64 [64,64] & 64 [64,64] & 0 [0,0] & 40/40; 40/40\\
\bottomrule\end{tabular}
\caption{New target observations to the first qualified decision on fixed recorded frames. Each cell gives a median and a marginal 95\% binomial order-statistic interval over 40 paired orders; the difference is selective minus reset within an order, not a difference of marginal medians. Both arms pay the same source audit, use a common target priority order, and create their own fixed-complement owner. Months are separate consumers. Before-census counts exclude the already retained points. Intervals quantify random-order variation conditional on these frames, not deployment risk or a population-level retention effect. Timing and all query-level outcomes remain in the release.}\label{tab:supplementary}
\end{table}

%% file: tables/cache_controls.tex
\begin{table}[t]\centering\small
\begin{tabular}{lrrrr}\toprule
Frame & \shortstack{Source census\\size} & \shortstack{Cached CI\\contrad. / qualified} & \shortstack{Blind points\\contrad. / qualified} & \shortstack{Tagged completion\\contrad. / qualified}\\\midrule
Heart & 303 & 90/401 & 90/401 & 0/263\\
C-MAPSS & 100 & 87/401 & 87/401 & 0/401\\
Seeds & 210 & 94/401 & 94/401 & 0/1\\
Air & 7,674 & 72/401 & 0/329 & 0/329\\
\bottomrule\end{tabular}
\caption{Unique threshold-query counts at the source-census endpoint of declared public-record consumer edits. Each frame has 401 fixed thresholds; the 40 source orders are identical at census and counted once. Partial 32- and 128-observation endpoints retain all 40 orders in the release. Cached intervals retain a source mean for a changed consumer; blind points relabel its samples; tagged completion applies qualified transformations or leaves unrelated target outcomes unknown. Contradictions mean disagreement with recorded truth, or exclusion of a permitted missing-outcome completion for Air. These query fractions are controlled edit diagnostics, not estimates of deployment failure probability. All methods receive the same source observations and public semantic information.}\label{tab:cache-controls}
\end{table}

%% file: main.bbl
\begin{thebibliography}{36}
\providecommand{\natexlab}[1]{#1}
\providecommand{\url}[1]{\texttt{#1}}
\expandafter\ifx\csname urlstyle\endcsname\relax
  \providecommand{\doi}[1]{doi: #1}\else
  \providecommand{\doi}{doi: \begingroup \urlstyle{rm}\Url}\fi

\bibitem[Adriansyah et~al.(2011)Adriansyah, van Dongen, and van~der
  Aalst]{adriansyah2011}
Arya Adriansyah, Boudewijn~F. van Dongen, and Wil M.~P. van~der Aalst.
\newblock Conformance checking using cost-based fitness analysis.
\newblock In \emph{2011 IEEE 15th International Enterprise Distributed Object
  Computing Conference}, pp.\  55--64, 2011.
\newblock \doi{10.1109/EDOC.2011.12}.

\bibitem[Alchourr{\'o}n et~al.(1985)Alchourr{\'o}n, G{\"a}rdenfors, and
  Makinson]{agm1985}
Carlos~E. Alchourr{\'o}n, Peter G{\"a}rdenfors, and David Makinson.
\newblock On the logic of theory change: Partial meet contraction and revision
  functions.
\newblock \emph{The Journal of Symbolic Logic}, 50\penalty0 (2):\penalty0
  510--530, 1985.
\newblock \doi{10.2307/2274239}.

\bibitem[Angelopoulos et~al.(2023)Angelopoulos, Bates, Fannjiang, Jordan, and
  Zrnic]{ppi2023}
Anastasios~N. Angelopoulos, Stephen Bates, Clara Fannjiang, Michael~I. Jordan,
  and Tijana Zrnic.
\newblock Prediction-powered inference.
\newblock \emph{Science}, 382:\penalty0 669--674, 2023.
\newblock \doi{10.1126/science.adi6000}.
\newblock URL \url{https://www.science.org/doi/10.1126/science.adi6000}.

\bibitem[Becker \& Kohavi(1996)Becker and Kohavi]{adult1996}
Barry Becker and Ron Kohavi.
\newblock Adult.
\newblock UCI Machine Learning Repository, 1996.
\newblock URL \url{https://archive.ics.uci.edu/dataset/2/adult}.

\bibitem[Berk et~al.(2013)Berk, Brown, Buja, Zhang, and Zhao]{berk2013}
Richard Berk, Lawrence Brown, Andreas Buja, Kai Zhang, and Linda Zhao.
\newblock Valid post-selection inference.
\newblock \emph{The Annals of Statistics}, 41\penalty0 (2):\penalty0 802--837,
  2013.
\newblock \doi{10.1214/12-AOS1077}.

\bibitem[Berti et~al.(2023)Berti, van Zelst, and Schuster]{pm4py2023}
Alessandro Berti, Sebastiaan van Zelst, and Daniel Schuster.
\newblock {PM4Py}: A process mining library for python.
\newblock \emph{Software Impacts}, 17:\penalty0 100556, 2023.
\newblock \doi{10.1016/j.simpa.2023.100556}.

\bibitem[Calinescu et~al.(2018)Calinescu, Weyns, Gerasimou, Iftikhar, Habli,
  and Kelly]{entrust2018}
Radu Calinescu, Danny Weyns, Simos Gerasimou, Muhammad~Usman Iftikhar, Ibrahim
  Habli, and Tim Kelly.
\newblock Engineering trustworthy self-adaptive software with dynamic assurance
  cases.
\newblock \emph{IEEE Transactions on Software Engineering}, 44\penalty0
  (11):\penalty0 1039--1069, 2018.
\newblock \doi{10.1109/TSE.2017.2738640}.
\newblock URL \url{https://doi.org/10.1109/TSE.2017.2738640}.

\bibitem[Charytanowicz et~al.(2010)Charytanowicz, Niewczas, Kulczycki,
  Kowalski, and Lukasik]{seeds2010}
Magorzata Charytanowicz, Jerzy Niewczas, Piotr Kulczycki, Piotr Kowalski, and
  Szymon Lukasik.
\newblock Seeds.
\newblock UCI Machine Learning Repository, 2010.
\newblock URL \url{https://archive.ics.uci.edu/dataset/236/seeds}.

\bibitem[Chugg et~al.(2023)Chugg, Wang, and Ramdas]{chugg2023}
Ben Chugg, Hongjian Wang, and Aaditya Ramdas.
\newblock A unified recipe for deriving (time-uniform) {PAC-Bayes} bounds.
\newblock \emph{Journal of Machine Learning Research}, 24\penalty0
  (372):\penalty0 1--61, 2023.
\newblock URL \url{https://jmlr.org/papers/v24/23-0401.html}.

\bibitem[Clopper \& Pearson(1934)Clopper and Pearson]{clopper1934}
C.~J. Clopper and E.~S. Pearson.
\newblock The use of confidence or fiducial limits illustrated in the case of
  the binomial.
\newblock \emph{Biometrika}, 26\penalty0 (4):\penalty0 404--413, 1934.
\newblock \doi{10.1093/biomet/26.4.404}.

\bibitem[Csillag et~al.(2025)Csillag, Struchiner, and Goedert]{ppe2025}
Daniel Csillag, Claudio~Jos{\'e} Struchiner, and Guilherme~Tegoni Goedert.
\newblock Prediction-powered e-values.
\newblock In \emph{Proceedings of the 42nd International Conference on Machine
  Learning}, volume 267 of \emph{Proceedings of Machine Learning Research},
  pp.\  11493--11514, 2025.
\newblock URL \url{https://proceedings.mlr.press/v267/csillag25a.html}.

\bibitem[de~Kleer(1986)]{dekleer1986}
Johan de~Kleer.
\newblock An assumption-based {TMS}.
\newblock \emph{Artificial Intelligence}, 28\penalty0 (2):\penalty0 127--162,
  1986.
\newblock \doi{10.1016/0004-3702(86)90080-9}.
\newblock URL
  \url{https://www.dekleer.org/Publications/An\%20Assumption-Based\%20TMS.pdf}.

\bibitem[Dwork et~al.(2015)Dwork, Feldman, Hardt, Pitassi, Reingold, and
  Roth]{dwork2015}
Cynthia Dwork, Vitaly Feldman, Moritz Hardt, Toniann Pitassi, Omer Reingold,
  and Aaron Roth.
\newblock Generalization in adaptive data analysis and holdout reuse.
\newblock In \emph{Advances in Neural Information Processing Systems},
  volume~28, 2015.
\newblock URL
  \url{https://proceedings.neurips.cc/paper/2015/hash/bad5f33780c42f2588878a9d07405083-Abstract.html}.

\bibitem[Dziugaite \& Roy(2018)Dziugaite and Roy]{dziugaite2018}
Gintare~Karolina Dziugaite and Daniel~M. Roy.
\newblock Data-dependent {PAC-Bayes} priors via differential privacy.
\newblock In \emph{Advances in Neural Information Processing Systems},
  volume~31, 2018.
\newblock URL
  \url{https://papers.nips.cc/paper_files/paper/2018/hash/9a0ee0a9e7a42d2d69b8f86b3a0756b1-Abstract.html}.

\bibitem[Foster \& Stine(2008)Foster and Stine]{foster2008}
Dean~P. Foster and Robert~A. Stine.
\newblock Alpha-investing: A procedure for sequential control of expected false
  discoveries.
\newblock \emph{Journal of the Royal Statistical Society: Series B (Statistical
  Methodology)}, 70\penalty0 (2):\penalty0 429--444, 2008.
\newblock \doi{10.1111/j.1467-9868.2007.00643.x}.

\bibitem[Gaur et~al.(2021)Gaur, Dang, Bhattacharya, and Bedathur]{gaur2021}
Garima Gaur, Abhishek Dang, Arnab Bhattacharya, and Srikanta Bedathur.
\newblock Computing and maintaining provenance of query result probabilities in
  uncertain knowledge graphs.
\newblock In \emph{Proceedings of the 30th ACM International Conference on
  Information and Knowledge Management}, pp.\  545--554, 2021.
\newblock \doi{10.1145/3459637.3482330}.
\newblock URL \url{https://arxiv.org/abs/2108.07758}.

\bibitem[Golovin et~al.(2010)Golovin, Krause, and Ray]{golovin2010}
Daniel Golovin, Andreas Krause, and Debajyoti Ray.
\newblock Near-optimal {Bayesian} active learning with noisy observations.
\newblock In \emph{Advances in Neural Information Processing Systems},
  volume~23, pp.\  766--774, 2010.
\newblock URL
  \url{https://proceedings.neurips.cc/paper_files/paper/2010/file/1e6e0a04d20f50967c64dac2d639a577-Paper.pdf}.

\bibitem[Green et~al.(2007)Green, Karvounarakis, and Tannen]{green2007}
Todd~J. Green, Grigoris Karvounarakis, and Val Tannen.
\newblock Provenance semirings.
\newblock In \emph{Proceedings of the 26th ACM SIGMOD-SIGACT-SIGART Symposium
  on Principles of Database Systems}, pp.\  31--40, 2007.
\newblock \doi{10.1145/1265530.1265535}.
\newblock URL \url{https://doi.org/10.1145/1265530.1265535}.

\bibitem[Hansson(2007)]{hansson2007}
Sven~Ove Hansson.
\newblock Three approaches to finitude in belief change.
\newblock In Toni R{\o}nnow-Rasmussen, Bj{\"o}rn Petersson, Jonas Josefsson,
  and Dan Egonsson (eds.), \emph{Hommage {\`a} Wlodek: Philosophical Papers
  Dedicated to Wlodek Rabinowicz}. Department of Philosophy, Lund University,
  2007.
\newblock URL
  \url{https://www.fil.lu.se/hommageawlodek/site/papper/HanssonSvenove.pdf}.

\bibitem[Horowitz \& Manski(2000)Horowitz and Manski]{horowitz2000}
Joel~L. Horowitz and Charles~F. Manski.
\newblock Nonparametric analysis of randomized experiments with missing
  covariate and outcome data.
\newblock \emph{Journal of the American Statistical Association}, 95\penalty0
  (449):\penalty0 77--84, 2000.
\newblock \doi{10.1080/01621459.2000.10473902}.
\newblock URL \url{https://doi.org/10.1080/01621459.2000.10473902}.

\bibitem[Imieli{\'n}ski \& Lipski(1984)Imieli{\'n}ski and
  Lipski]{imielinski1984}
Tomasz Imieli{\'n}ski and Witold Lipski.
\newblock Incomplete information in relational databases.
\newblock \emph{Journal of the ACM}, 31\penalty0 (4):\penalty0 761--791, 1984.
\newblock \doi{10.1145/1634.1886}.

\bibitem[Janosi et~al.(1989)Janosi, Steinbrunn, Pfisterer, and
  Detrano]{heart1989}
Andras Janosi, William Steinbrunn, Matthias Pfisterer, and Robert Detrano.
\newblock Heart disease.
\newblock UCI Machine Learning Repository, 1989.
\newblock URL \url{https://archive.ics.uci.edu/dataset/45/heart+disease}.

\bibitem[Javdani et~al.(2014)Javdani, Chen, Karbasi, Krause, Bagnell, and
  Srinivasa]{javdani2014}
Shervin Javdani, Yuxin Chen, Amin Karbasi, Andreas Krause, Drew Bagnell, and
  Siddhartha Srinivasa.
\newblock Near optimal {Bayesian} active learning for decision making.
\newblock In \emph{Proceedings of the Seventeenth International Conference on
  Artificial Intelligence and Statistics}, volume~33 of \emph{Proceedings of
  Machine Learning Research}, pp.\  430--438, 2014.
\newblock URL \url{https://proceedings.mlr.press/v33/javdani14.html}.

\bibitem[Lee et~al.(2026)Lee, Han, Sen, Yeom, {\c{C}}etintemel, and
  Datta]{evergreen2026}
Alexander~W. Lee, Benjamin Han, Shayak Sen, Sam Yeom, Ugur {\c{C}}etintemel,
  and Anupam Datta.
\newblock Evergreen: Efficient claim verification for semantic aggregates.
\newblock \emph{arXiv preprint arXiv:2604.26180}, 2026.
\newblock URL \url{https://arxiv.org/html/2604.26180v2}.
\newblock Version 2.

\bibitem[Murata(1989)]{murata1989}
Tadao Murata.
\newblock Petri nets: Properties, analysis and applications.
\newblock \emph{Proceedings of the IEEE}, 77\penalty0 (4):\penalty0 541--580,
  1989.
\newblock \doi{10.1109/5.24143}.

\bibitem[{NASA}(2022)]{cmapss2008}
{NASA}.
\newblock {CMAPSS} jet engine simulated data.
\newblock NASA Open Data Portal, Prognostics Center of Excellence, 2022.
\newblock URL
  \url{https://data.nasa.gov/dataset/cmapss-jet-engine-simulated-data}.
\newblock Legacy simulation data; original construction described by Saxena et
  al. (2008).

\bibitem[Oveisi et~al.(2014)Oveisi, Delgrande, Pelletier, and
  Popowich]{oveisi2014}
Mehrdad Oveisi, James~P. Delgrande, Francis~Jeffry Pelletier, and Fred
  Popowich.
\newblock Belief change and base dependence.
\newblock In \emph{Principles of Knowledge Representation and Reasoning}, pp.\
  151--159, 2014.
\newblock URL \url{https://www2.cs.sfu.ca/~jim/publications/KR14b.pdf}.

\bibitem[Rivasplata et~al.(2020)Rivasplata, Kuzborskij, Szepesv{\'a}ri, and
  Shawe-Taylor]{rivasplata2020}
Omar Rivasplata, Ilja Kuzborskij, Csaba Szepesv{\'a}ri, and John Shawe-Taylor.
\newblock {PAC-Bayes} analysis beyond the usual bounds.
\newblock In \emph{Advances in Neural Information Processing Systems},
  volume~33, 2020.
\newblock URL
  \url{https://proceedings.neurips.cc/paper_files/paper/2020/hash/c3992e9a68c5ae12bd18488bc579b30d-Abstract.html}.

\bibitem[Rodr{\'i}guez-G{\'a}lvez et~al.(2024)Rodr{\'i}guez-G{\'a}lvez,
  Thobaben, and Skoglund]{morepac2024}
Borja Rodr{\'i}guez-G{\'a}lvez, Ragnar Thobaben, and Mikael Skoglund.
\newblock More pac-bayes bounds: From bounded losses, to losses with general
  tail behaviors, to anytime validity.
\newblock \emph{Journal of Machine Learning Research}, 25\penalty0
  (110):\penalty0 1--43, 2024.
\newblock URL \url{https://jmlr.org/papers/v25/23-1360.html}.

\bibitem[Saxena et~al.(2008)Saxena, Goebel, Simon, and Eklund]{saxena2008}
Abhinav Saxena, Kai Goebel, Don Simon, and Neil Eklund.
\newblock Damage propagation modeling for aircraft engine run-to-failure
  simulation.
\newblock In \emph{2008 International Conference on Prognostics and Health
  Management}, pp.\  1--9, 2008.
\newblock \doi{10.1109/PHM.2008.4711414}.
\newblock URL \url{https://ntrs.nasa.gov/citations/20090029214}.

\bibitem[van Dongen(2012)]{bpi2012}
B.~F. van Dongen.
\newblock {BPI} challenge 2012.
\newblock Eindhoven University of Technology dataset, 2012.
\newblock URL \url{https://research.tue.nl/en/datasets/bpi-challenge-2012/}.

\bibitem[Vito(2008)]{air2008}
Saverio Vito.
\newblock Air quality.
\newblock UCI Machine Learning Repository, 2008.
\newblock URL \url{https://archive.ics.uci.edu/dataset/360/air+quality}.

\bibitem[Vovk \& Wang(2021)Vovk and Wang]{vovk2021}
Vladimir Vovk and Ruodu Wang.
\newblock E-values: Calibration, combination and applications.
\newblock \emph{The Annals of Statistics}, 49\penalty0 (3):\penalty0
  1736--1754, 2021.
\newblock \doi{10.1214/20-AOS2020}.

\bibitem[Waudby-Smith \& Ramdas(2020)Waudby-Smith and Ramdas]{wor2020}
Ian Waudby-Smith and Aaditya Ramdas.
\newblock Confidence sequences for sampling without replacement.
\newblock In \emph{Advances in Neural Information Processing Systems},
  volume~33, 2020.
\newblock URL
  \url{https://papers.nips.cc/paper_files/paper/2020/file/e96c7de8f6390b1e6c71556e4e0a4959-Paper.pdf}.

\bibitem[Wu et~al.(2024)Wu, Zhang, Ch{\'e}rief-Abdellatif, and
  Seldin]{recursive2024}
Yi-Shan Wu, Yijie Zhang, Badr-Eddine Ch{\'e}rief-Abdellatif, and Yevgeny
  Seldin.
\newblock Recursive {PAC-Bayes}: A frequentist approach to sequential prior
  updates with no information loss.
\newblock In \emph{Advances in Neural Information Processing Systems},
  volume~37, 2024.
\newblock URL
  \url{https://proceedings.neurips.cc/paper_files/paper/2024/hash/1feddcfb42229ca84d3070c7d540daaa-Abstract-Conference.html}.

\bibitem[Zhang et~al.(2026)Zhang, Cai, Yu, and Simeone]{pprm2026}
Guangyi Zhang, Yunlong Cai, Guanding Yu, and Osvaldo Simeone.
\newblock Prediction-powered risk monitoring of deployed models for detecting
  harmful distribution shifts.
\newblock In \emph{Proceedings of the 43rd International Conference on Machine
  Learning}, volume 306 of \emph{Proceedings of Machine Learning Research},
  pp.\  154241--154262, 2026.
\newblock URL \url{https://proceedings.mlr.press/v306/zhang26g.html}.

\end{thebibliography}
